\documentclass[11pt,english,onecolumn, draftcls]{IEEEtran}
\usepackage[T1]{fontenc}
\usepackage[latin9]{inputenc}
\usepackage{color}
\usepackage{units}
\usepackage{amsmath}
\usepackage{amssymb}
\usepackage{graphicx}
\usepackage{esint}

\makeatletter

\providecommand{\tabularnewline}{\\}

  \newtheorem{assumption}{Assumption}
  \newtheorem{definitn}{Definition}
  \newtheorem{remrk}{Remark}
  \newtheorem{lemma}{Lemma}
  \newtheorem{thm}{Theorem}
  \newtheorem{problem}{Problem}
  \newtheorem{cor}{Corollary}

\usepackage{cite}

\makeatother

\usepackage{babel}
\begin{document}

\title{{\huge Limited Feedback Design for Interference Alignment on MIMO
Interference Networks with Heterogeneous Path Loss and Spatial Correlations}}

\author{Xiongbin Rao, Liangzhong Ruan, \emph{Student Member, IEEE}, and \\Vincent
K.N. Lau, \emph{Fellow, IEEE}%
\thanks{The authors are with ECE Department, the Hong Kong University of Science
and Technology, Hong Kong (e-mails: \{xrao,stevenr,eeknlau\}@ust.hk).%
}\vspace{-0.6in}
}
\maketitle
\begin{abstract}
Interference alignment is degree of freedom optimal in $K$-user MIMO
interference channels and many previous works have studied the transceiver
designs. However, these works predominantly focus on networks with
perfect channel state information at the transmitters and symmetrical
interference topology. In this paper, we consider a limited feedback
system with \textit{heterogeneous path loss and spatial correlations},
and investigate how the dynamics of the interference topology can
be exploited to improve the feedback efficiency. We propose a novel
spatial codebook design, and perform dynamic quantization via bit
allocations to adapt to the asymmetry of the interference topology.
We bound the system throughput under the proposed dynamic scheme in
terms of the transmit SNR, feedback bits and the interference topology
parameters. It is shown that when the number of feedback bits scales
with SNR as $C_{s}\cdot\log\textrm{SNR}$, the sum degrees of freedom
of the network are preserved. Moreover, the value of scaling coefficient
$C_{s}$ can be significantly reduced in networks with asymmetric
interference topology.
\end{abstract}

\section{Introduction}

\subsection{Prior Works}

The capacity region for the interference channel remains unknown,
although researchers have been working on it for more than thirty
years \cite{cadambe2008interference,jafar2010MNDoF}. Conventional
schemes either treat interference as noise or use channel orthogonalization
to avoid interference. However, these schemes are non-capacity achieving
in general. Interference alignment (IA), which tries to align the
aggregate interference from different transmitters (Txs) into a lower
dimensional subspace at each receiver (Rx), is shown to be degree
of freedom (DoF) optimal in interference channels \cite{cadambe2008interference}
as well as other network scenarios such as the MIMO-X channels \cite{jafar2007MIMOX}.
In addition, despite the fact that IA is optimal only at high SNR,
the IA method potentially gives simpler solutions because the direct
channels are not needed to compute the precoders and decorrelators
\cite{gomadam2011distributed}. As such, there is a surge in the research
interests of IA.

To implement IA, signal dimensions are needed and those dimensions
can be obtained via symbol extension (time or frequency domain) or
by multiple antennas (spatial domain) \cite{cadambe2008interference,jafar2010MNDoF}.
Existing IA deisgn involving symbol extensions has high signal dimensions%
\footnote{The IA solution exploiting symbol extensions \cite{cadambe2008interference,jafar2010MNDoF}
requires $\mathcal{O}((KN)^{2K^{2}N^{2}})$ ($K$ denotes the number
of Tx-Rx pair, $N$ the number of antennas at each node) dimension
of signal space to achieve the optimal DoF, which is difficult to
realize in practice.%
} \cite{jafar2010MNDoF} and is hard to implement in practice. As a
result, many recent IA works have considered IA solutions in the spatial
domain, i.e., without symbol extensions \cite{gomadam2011distributed,peters2011cooperative,santamaria2010maximum}.
However, these approaches are all based on the assumption of perfect
channel state information at the transmitters (CSIT), which is hard
to obtain in practice. As such, we shall focus on studying the limited
feedback design and the associated performance analysis of IA.

The issue of limited feedback on MIMO networks is widely studied in
the research community. For instance, for MIMO broadcast channels
(BC) with zero-forcing beamforming, the performance loss due to limited
feedback is studied in \cite{jindal2006mimo,yoo2007multi}. However,
these works cannot be easily extended to MIMO interference channels
with IA processing as the analysis highly depends on the BC topology
and the zero-forcing strategy at Txs. There are some works that consider
MIMO interference networks adopting IA under limited feedback. For
instance, in \cite{ayach2012interference}, IA with analog feedback
is considered and the performance degradation is studied. In \cite{kim2012new},
a new quantization scheme is studied to reduce the quantization distortion
on MIMO interference networks. However, these works have considered
homogeneous path loss and i.i.d fading and thus failed to exploit
the potential benefits introduced by the asymmetric interference topology.
Besides these works, there are also some papers \cite{krishnamachari2009interference,thukral2009interference}
that investigate the feedback bits scaling law on MIMO interference
networks. The authors show that it is sufficient to maintain the maximum
DoF feasible by IA when the number of CSI feedback bits at each Rx
node scales on $\mathcal{O}(\log(\textrm{SNR}))$. However, these
works analyze the scaling law in the high feedback bits regime only
and thus fail to quantize the network performance when we have finite
feedback bits. Moreover, the potential possibility brought by heterogeneous
path loss and spatial correlations to reduce the scaling bits are
not explored.

\subsection{Remaining Challenges}

In this paper, we consider MIMO interference networks with heterogeneous
path loss as well as \textit{\emph{spatial}} correlations, and focus
on investigating the limited feedback performance of IA in spatial
domain. In view of the prior works, there are two key technical challenges
that need to be addressed.
\begin{itemize}
\item \textbf{How to exploit heterogeneous path loss and spatial correlations
to reduce the limited feedback?} Traditionally, the CSI matrices are
stacked into long vectors and then quantized by regular vector quantization
(VQ) \cite{love2004value,jongren2002combining}. However, as these
schemes adopt symmetric codebooks, they are inefficient when the channel
matrices are spatially correlated \cite{shiu2000fading,ruan2011dynamic}
or have heterogeneous path loss. Intuitively, when the links in the
MIMO interference networks have spatial correlations, the normalized
channel matrices will no longer be isotropic over the Grassmanian
subspace \cite{dai2008quantization}. Furthermore, if the links in
the MIMO interference network have different path loss and spatial
correlations, they should not be allocated the same amount of bits
for limited feedback. The challenge is therefore how to exploit this
asymmetry in the network topology to improve the efficiency of limited
feedback. There are some works on point-to-point MIMO links that exploit
the spatial correlations to improve limited feedback performance \cite{raghavan2007systematic,raghavan2006near}.
However, these works require closed-form precoders and hence, they
cannot be extended to our problem with general MIMO interference network
topology where there is no closed-form IA transceivers. In this paper,
we propose a novel spatial codebook design to exploit the spatial
correlations on MIMO interference networks. There are some works that
consider dynamic bits allocations, such as the feedback bits partitioning
between desired and interfering channels \cite{bhagavatula2011adaptive},
to improve the feedback efficiency. Motivated by this idea, we further
perform dynamic quantization via bit allocations for different interference
links to further exploit the asymmetry of the interference topology.
\item \textbf{What is the trade-off between the feedback rate and the network
throughput in general asymmetric MIMO interference networks?} In literature,
there are very limited works that analyze the performance loss due
to limited feedback for IA on MIMO interference networks. In \cite{cho2011feedback},
the author gives some analysis on the trade-off between the throughput
loss and limited feedback rate. However, the approach in this work
relies on the closed-form IA solution for the precoder and hence,
only the specialized topology and single stream transmission case
is analyzed. Yet, the approach cannot be extended to our case because
of the lack of closed-form IA solution for precoders. In this paper,
we shall study the tradeoff between the network throughput and feedback
rate based on the proposed scheme for more general interference topologies.
From the analysis, we can obtain useful insights on how the system
performance depends on parameters of the network topology such as
the path loss and the spatial correlations. 
\end{itemize}

\subsection{Outline of the Paper}

This paper is organized as follows. In Section II, we give our system
model of $K$-user MIMO interference networks with heterogeneous path
loss as well as spatial correlations, and specify the limited feedback
topology. In Section III, we shall address the first technical challenge.
We first illustrate the potential advantage of heterogeneous path
loss and spatial correlations on the improvement of feedback efficiency
using a toy example. Based on that, we then propose a novel spatial
codebook design as well as dynamic quantization via bit allocations
to adapt to the interference topology. In Section IV, we shall address
the second challenge. We analyze network throughput under the proposed
limited feedback scheme and give the throughput bounds. In Section
V, we compare the performances of the proposed dynamic feedback scheme
with several baselines via simulations. Through both analysis and
simulations, we show that by exploiting the heterogeneous path loss
and spatial correlations in the MIMO interference network, the proposed
scheme significantly improves the system performance in a wide range
of operation regimes.

\textit{Notation}s: The following notations are used in the paper.
Uppercase and lowercase boldface denote matrices and vectors respectively.
The operators $(\cdot)^{*}$, $(\cdot)^{T}$, $(\cdot)^{H}$, $||\cdot||$,
vec$(\cdot)$, $\mathbb{E}\{\cdot\}$, Tr$(\cdot)$, $\textrm{rank}(\cdot)$,
$\otimes$, are complex conjugate, transpose, conjugate transpose,
Frobenius norm, stacking vectorization, expectation, trace, rank,
Kronecker product operator respectively, diag$(\cdot)$ denotes forming
matrix operator using the inputs as diagonal blocks, $\textrm{span}(\{\mathbf{a}\})$
denotes the linear space spanned by the vectors in $\{\mathbf{a}\}$,
$\log(\cdot)$ is the logarithm of base 2, $\mathcal{O}(\cdot)$ denotes
the asymptotic upper bound, i.e., $f(x)=\mathcal{O}(g(x))$ if there
is a positive constant $M$ such that $|f(x)|\leq M\cdot|g(x)|$ for
all sufficiently large $x$.

\section{System Model}

In this section, we shall first elaborate the interference network
topology with heterogeneous path loss and spatial correlation. We
further define the notion of interference topology profile and illustrate
by using some examples. Finally, we will elaborate the limited feedback
topology for the MIMO interference network with IA processing.

\subsection{Topology of the MIMO Interference Network}

We consider a $K$-user MIMO interference network in which each Tx
is equipped with $N_{t}$ antennas and each Rx with $N_{r}$ antennas
as shown in Fig. \ref{fig:Information-flow-for} (A). Denote the transmit
SNR at each Tx as $P$, the large scale fading gain from Tx $i$ to
Rx $j$ as $l_{ji}$, the small scale fading matrix from Tx $i$ to
Rx $j$ as $\mathbf{H}_{ji}\in\mathbb{C}^{N_{r}\times N_{t}}$. Let
$d$ be the number of data streams transmitted by each Tx-Rx pair.
The received signal $\mathbf{y}_{j}\in\mathbb{C}^{d\times1}$ at the
Rx $j$ is given by:

\begin{equation}
\mathbf{y}_{j}=l_{jj}^{\nicefrac{1}{2}}\mathbf{\mathbf{U}}_{j}^{H}\mathbf{H}_{jj}\mathbf{V}_{j}\mathbf{x}_{j}+\mathbf{\mathbf{U}}_{j}^{H}(\sum_{i\neq j}^{K}l_{ji}^{\nicefrac{1}{2}}\mathbf{H}_{ji}\mathbf{V}_{i}\mathbf{x}_{i}+\mathbf{z}_{j}),\qquad\forall j\in\left\{ 1,2,\cdots K\right\} \label{eq:signal}
\end{equation}
where $\mathbf{x}_{i}\sim\mathcal{CN}(\mathbf{0},\;\frac{P}{d}\mathbf{I}_{d})$
is the encoded information symbol at Tx $i$ for corresponding Rx
$i$, $\mathbf{V}_{i}\in\mathbb{C}^{N_{t}\times d}$ the transmit
precoding matrix of Tx $i$, $\mathbf{U}_{j}\in\mathbb{C}^{N_{r}\times d}$
the decorrelator of Rx $j$, and $\mathbf{z}_{j}\in\mathbb{C}^{N_{r}\times1}$
the complex Gaussian noise with zero mean and unit variance. We have
the following assumption regarding $\mathbf{H}_{ji}$ by using the
Kronecker correlation model \cite{tulino2004random}.
\begin{assumption}
[Channel Fading Model]The channel matrix $\mathbf{H}_{ji}$ in this
paper is given by:
\begin{equation}
\mathbf{H}_{ji}=(\mathbf{\Phi}_{ji}^{r})^{\nicefrac{1}{2}}\mathbf{H}_{ji}^{w}(\mathbf{\Phi}_{ji}^{t})^{\nicefrac{1}{2}}\label{eq:csi}
\end{equation}
where%
\footnote{We define the square root of a PSD matrix $\mathbf{\Phi}$ as $\mathbf{\Phi}^{\frac{1}{2}}=\mathbf{F}\mathbf{\Lambda}^{\frac{1}{2}}\mathbf{F}^{H}$,
where $\mathbf{\Phi}=\mathbf{F}\mathbf{\Lambda}\mathbf{F}^{H}$ denotes
the eigenvalue decomposition.%
} $\mathbf{H}_{ji}^{w}\in\mathbb{C}^{N_{r}\times N_{t}}$ and each
entry of it is i.i.d. $\mathcal{CN}(0,1)$, $\mathbf{\Phi}_{ji}^{r}\in\mathbb{C}^{N_{r}\times N_{r}}$,
$\mathbf{\Phi}_{ji}^{t}\in\mathbb{C}^{N_{t}\times N_{t}}$ are deterministic
positive semi-definite (PSD) matrices which stand for the spatial
correlation matrices at Rx, Tx side respectively, $\mathbf{\Phi}_{ji}^{r}$,
$\mathbf{\Phi}_{ji}^{t}$ are normalized such that $\textrm{Tr}(\mathbf{\Phi}_{ji}^{r})=N_{r}$,
$\textrm{Tr}(\mathbf{\Phi}_{ji}^{t})=N_{t}$. Denote $M_{ji}^{r}=\textrm{rank}(\mathbf{\Phi}_{ji}^{r})$,
$M_{ji}^{t}=\textrm{rank}(\mathbf{\Phi}_{ji}^{t})$ ($0<M_{ji}^{r}\leq N_{r}$,
$0<M_{ji}^{t}\leq N_{t}$), and the non-zero eigenvalues of $\mathbf{\Phi}_{ji}^{r}$,
$\mathbf{\Phi}_{ji}^{t}$ as $\{\lambda_{ji,1},\cdots\lambda_{ji,M_{ji}^{r}}\}$,
$\{\sigma_{ji,1},\cdots\sigma_{ji,M_{ji}^{t}}\}$ respectively.\hfill \IEEEQED
\end{assumption}

\subsection{Interference Topology Profile}

In this section, we define the notion of $\emph{interference topology profile}$
($\mathcal{ITP}$) which is used to capture the heterogeneous path
loss and spatial correlations in the MIMO interference network. 
\begin{definitn}
[Interference Topology Profile]We define the set of all the channel
statistics $\mathcal{ITP}=\{\mathbf{\Phi}_{ji}^{r},\:\mathbf{\Phi}_{ji}^{t},\; l_{ji}\}$
as the \textit{interference topology profile}.\hfill \IEEEQED
\end{definitn}

As such, the $\mathcal{ITP}$ fully characterizes the heterogeneity
of the path loss and spatial correlation among the interference links.
We give several examples below with a ($K=4$, $N_{t}=3$, $N_{r}=2$,
$d=1$) interference network.
\begin{itemize}
\item \textbf{A fully connected MIMO interference network with i.i.d. Rayleigh
fading}: If $\mathbf{\Phi}_{ji}^{r}=\mathbf{\Phi}_{ji}^{t}=\mathbf{I}$,
$l_{ji}=1$, $\forall i,j\in\{1,\cdots4\}$, then the MIMO interference
network reduces to the conventional fully connected interference channel
in which all the elements of the channel matrices $\{\mathbf{H}_{ji}\}$
are i.i.d. Rayleigh fading.
\item \textbf{A fully connected MIMO interference network with asymmetric
spatial correlation}: Due to local scattering effects, the MIMO channel
matrix $\mathbf{H}_{ji}$ may not be i.i.d. and in some cases, there
will be spatial correlations. For instance, if $\mathbf{\Phi}_{31}^{t}=\textrm{diag}([\begin{array}{ccc}
2.8 & 0.1 & 0.1\end{array}])$, other $\mathbf{\Phi}_{ji}^{t}=\mathbf{I}$, all $\mathbf{\Phi}_{ji}^{r}=\mathbf{I}$,
$l_{ji}=1$ in $\mathcal{ITP}$, then this network corresponds to
an example of a fully connected MIMO interference network with asymmetric
spatial correlation.
\item \textbf{A partially connected MIMO interference network with heterogeneous
path loss}: In practice, different cross-links might have heterogeneous
path loss due to different geometric distributions between Txs and
Rxs. For instance, if $l_{14}=10^{-8}$, other $l_{ji}=1$, all $\mathbf{\Phi}_{ji}^{t}=\mathbf{\Phi}_{ji}^{r}=\mathbf{I}$,
then this network corresponds to an example of a partially connected
(since $l_{14}\ll\textrm{other }l_{ji}$, the link from Tx 4 to Rx
1 can be regarded as disconnected.) MIMO interference network with
heterogeneous path loss.
\end{itemize}
\begin{figure}
\begin{centering}
\includegraphics[scale=0.8]{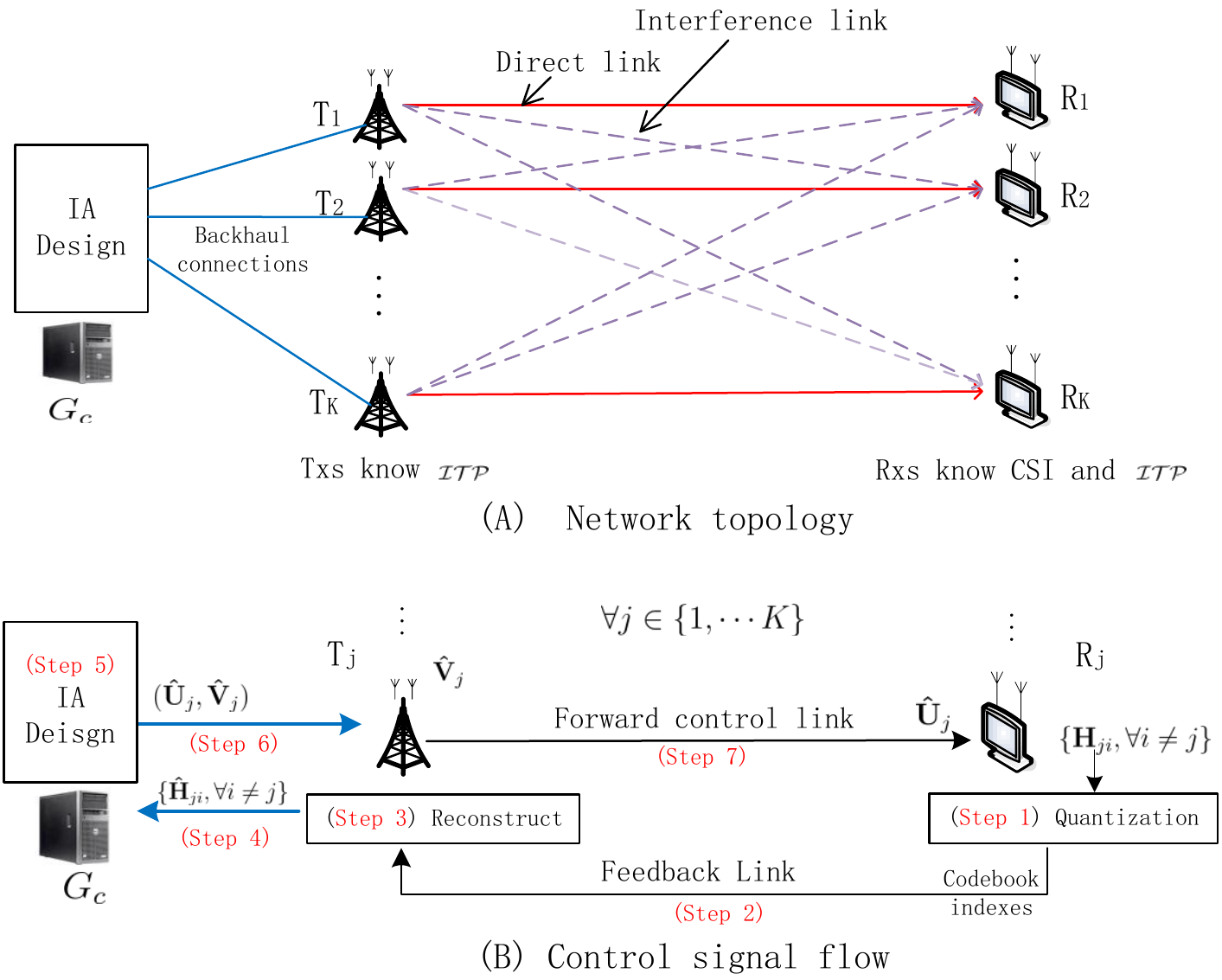}
\par\end{centering}

\noindent \caption{\label{fig:Information-flow-for}System model.}
\end{figure}

\subsection{Limited Feedback Topology }
\begin{assumption}
[Network Connection and Information State]As illustrated in Fig.
\ref{fig:Information-flow-for} (A), we assume that there is a BS
controller $G_{c}$ that has backhaul connections to all the Txs.
We also assume that the instantaneous CSI is available at the corresponding
Rx side but not at the Tx side, and that the \textit{$\{\mathbf{\Phi}_{ji}^{r},\:\mathbf{\Phi}_{ji}^{t},\; l_{ji},\forall i\}$}
is available at both the Rx $j$ and Tx $j$. \hfill \IEEEQED
\end{assumption}

\begin{remrk}
[Practical Considerations]In practice, $\{\mathbf{\Phi}_{ji}^{r},\:\mathbf{\Phi}_{ji}^{t},\; l_{ji},\forall i\}$
can be obtained at Tx by either reciprocity of path loss and spatial
correlations or explicit feedback of them from the Rx side. In either
case, the path loss and spatial correlations are slowly varying and
the acquisition of them at the transmitters can be done with very
small overhead compared with instantaneous CSI feedback.
\end{remrk}

In this paper, we deploy the IA algorithm with iterative interference
leakage minimization in \cite{gomadam2011distributed} to compute
the transceivers. As all the CSIs are collected in $G_{c}$, we shall
implement the IA algorithm in a centralized manner, such that there
will be no over-the-air iterations among the nodes (over-the air iterations
will consume excessive signaling overhead and backhaul bandwidth).
On the other hand, we shall focus on the feedback scheme for the \textit{cross
links} only as the IA algorithm \cite{gomadam2011distributed} is
only related to the cross links. The outline of limited feedback topology
is described in the following algorithm, and is also illustrated in
Fig. \ref{fig:Information-flow-for} (B). 

\textit{Algorithm 1 (Limited Feedback Topology for MIMO Interference
Network Adopting IA Processing):}
\begin{itemize}
\item \textbf{Step 1}: At each Rx $j$, the cross link CSI $\{\mathbf{H}_{ji},\,\forall i\neq j\}$
are quantized to be $\{\mathbf{\hat{H}}_{ji},\,\forall i\neq j\}$
using the spatial codebooks $\{\mathbb{\mathcal{C}}_{ji},\,\forall i\neq j\}$
with $\{B_{ji},\,\forall i\neq j\}$ bits respectively. 
\item \textbf{Step 2}: The quantized codeword indexes are then fedback to
the $j$-th Tx using feedback link.
\item \textbf{Step} \textbf{3}: Each Tx $j$ receives the codebook indexes
and reconstructs the CSIs to be $\{\mathbf{\hat{H}}_{ji},\,\forall i\neq j\}$.
\item \textbf{Step 4}: The Txs forward the reconstructed CSIs to $G_{c}$
through the backhaul link.
\item \textbf{Step 5}: Based on the collected quantized CSIs from all Txs,
$G_{c}$ computes the IA transceivers as 
\begin{equation}
\left\{ (\mathbf{\hat{U}}_{i},\mathbf{\hat{V}}_{i})\right\} =\textrm{IA}\left(\{\mathbf{\hat{H}}_{nk},\forall n,k,\; n\neq k\}\right)\label{eq:IA_processing}
\end{equation}
where IA denotes the IA processing \cite{gomadam2011distributed}
and $(\mathbf{\hat{U}}_{i},\mathbf{\hat{V}}_{i})$ denotes the designed
IA transceiver. 
\item \textbf{Step 6}: $G_{c}$ distributes$\{(\mathbf{\hat{U}}_{i},\mathbf{\hat{V}}_{i})\}$
to the Txs.
\item \textbf{Step 7}: Each Tx $i$ forwards $\mathbf{\hat{U}}_{i}$ to
the Rx $i$ using the forward control link. \hfill \IEEEQED
\end{itemize}

The design of the spatial codebook and bit allocations $\{\mathbb{\mathcal{C}}_{ji},B_{ji}\; i\neq j\}$
mentioned in Algorithm 1 will be discussed in detail in Section III.
Note that these designs are adaptive to the path loss and spatial
correlations which are long term statistics. Hence, once the spatial
codebooks and bit allocations are determined, each Rx will quantize
the instantaneous CSIs independently using the corresponding codebooks
and allocated bits. 

Assume that the network is IA feasible, we have the following properties
about $\left\{ (\mathbf{\hat{U}}_{i},\mathbf{\hat{V}}_{i})\right\} $
\cite{gomadam2011distributed}, 
\begin{equation}
\mathbf{\hat{U}}_{i}^{H}\mathbf{\hat{U}}_{i}=\mathbf{I}_{d\times d},\;\mathbf{\hat{V}}_{i}^{H}\mathbf{\hat{V}}_{i}=\mathbf{I}_{d\times d},\quad\forall i\in\{1,\cdots K\}.\label{eq:IA_uv_pro}
\end{equation}
\begin{equation}
\mathbf{\hat{U}}_{j}^{H}\mathbf{\hat{H}}_{ji}\mathbf{\hat{V}}_{i}=\mathbf{0}_{d\times d},\quad\forall i\neq j,\; i,\, j\in\{1,\cdots K\}.\label{eq:IA_UV_property}
\end{equation}

Due to the limited feedback CSI, the IA transceiver cannot achieve
perfect alignment and thus there will be some residual interference.
Denote the residual interference to noise ratio (RINR) at Rx $j$
as $I_{j}$, we have 
\begin{equation}
I_{j}=\dfrac{P}{d}\sum_{i,i\neq j}^{K}l_{ji}||\mathbf{\hat{U}}_{j}^{H}\mathbf{H}_{ji}\mathbf{\hat{V}}_{i}||^{2}.\label{eq:Residue_interference_expression}
\end{equation}

\section{Limited Feedback with Dynamic Quantization}

Define the dynamic feedback policy for the network as 
\begin{equation}
\mathcal{\mathcal{P}}=\{\mathcal{C}_{ji},\, B_{ji},\, i\neq j,\,\forall i,\, j=1\cdots K\},\label{eq:multi-resolution_codebook}
\end{equation}
where $\mathcal{C}_{ji}$ is the codebook for link $\mathbf{H}_{ji}$,
and $B_{ji}$ denotes the bits allocated for $\mathcal{C}_{ji}$.
We investigate in this section how the limited feedback scheme $\mathcal{P}$
is designed to adapt to the $\mathcal{ITP}$ in the MIMO interference
network. We first illustrate the motivation of dynamic quantization
according to the $\mathcal{ITP}$ based on a toy example. We then
elaborate the details of the proposed feedback scheme $\mathcal{P}$,
which is divided into two parts, namely the \textit{spatial codebook
design} in Section III-B and \textit{dynamic quantization via bit
allocations} in Section III-C.

\subsection{Motivation}

Consider a $K=4$, $N_{t}=3$, $N_{r}=2$, $d=1$ interference network
where $\mathbf{\Phi}_{13}^{t}=\textrm{diag}([\begin{array}{ccc}
2.8 & 0.1 & 0.1\end{array}])$, $l_{13}=1$, $\mathbf{\Phi}_{14}^{t}=\textrm{diag}([\begin{array}{ccc}
1 & 1 & 1\end{array}])$, $l_{14}=0.1$ (all other $\mathbf{\Phi}_{ji}^{r}=\mathbf{\Phi}_{ji}^{t}=\mathbf{I}$,
$l_{ji}=1$), transmit SNR $P=1$. Assume that all the other channel
matrices are perfectly known by the BS controller $G_{c}$ except
for $\mathbf{H}_{13}$ and $\mathbf{H}_{14}$. We shall then investigate
below the feedback scheme for these two links only. As the fading
gain from Tx $3$ to Rx $1$ is much larger ($l_{13}\gg l_{14}$),
it is probable that better performance may be achieved if we concentrate
on quantizing $\mathbf{H}_{13}$ only. Moreover, as $\mathbf{H}_{13}=\mathbf{H}_{13}^{w}((\mathbf{\Phi}_{13}^{t})^{\frac{1}{2}})=\mathbf{H}_{13}^{w}\cdot\textrm{diag}([\begin{array}{ccc}
2.8 & 0.1 & 0.1\end{array}])^{0.5}$ according to (\ref{eq:csi}) and thus the columns of $\mathbf{H}_{13}$
have different gains. It is probable that better performance can be
achieved if the codebook to quantize $\mathbf{H}_{13}$ has the same
statistic distribution as $\mathbf{H}_{13}$. To verify this hypothesis,
we compare the performance of the following two quantization schemes.
Note the dynamic quantization scheme illustrated below is only a simple
toy scheme which helps to show the idea. In the following two schemes,
$\mathcal{C}_{0}$ denotes a random vector quantization codebook \cite{jindal2006mimo}
with $(6\times1)$ codewords.
\begin{itemize}
\item \textbf{Conventional VQ}: Allocate equal bits to $\mathbf{H}_{13},\mathbf{H}_{14}$
and use $\mathcal{C}_{0}$ to quantize $\textrm{vec}(\mathbf{H}_{13}),\,\textrm{vec}(\mathbf{H}_{14})$
\cite{jindal2006mimo}.
\item \textbf{Dynamic Quantization}: Use all the bits to quantize $\mathbf{H}_{13}$
only, the codebook to quantize $\textrm{vec}(\mathbf{H}_{13})$ is
given by $\mathcal{C}=\left\{ \mathbf{u}\mid\mathbf{u}=\frac{\textrm{diag}([\begin{array}{cccccc}
28 & 28 & 1 & 1 & 1 & 1\end{array}])^{0.5}\cdot\mathbf{f}}{||\textrm{diag}([\begin{array}{cccccc}
28 & 28 & 1 & 1 & 1 & 1\end{array}])^{0.5}\cdot\mathbf{f}||},\,\mathbf{f}\in\mathcal{C}_{0}\right\} $.
\end{itemize}
\begin{table}
\begin{centering}
\begin{tabular}{|c|c|c|c|}
\hline 
Sum feedback bits & 4 & 10 & 16\tabularnewline
\hline 
\hline 
 Conventional VQ & 0.9057 & 0.5826 & 0.3219\tabularnewline
\hline 
Dynamic Quantization & 0.3055 & 0.1595 & 0.1333\tabularnewline
\hline 
\end{tabular}
\par\end{centering}

\centering{}\caption{\label{tab:Residue-interference-comparison}Residue interference comparison
Versus bits}
\end{table}

The comparison of RINR at Rx 1 versus the sum feedback bits of the
two links is illustrated in Table \ref{tab:Residue-interference-comparison}.
We see that the\textit{ Dynamic Quantization} scheme can achieve much
lower RINR than the \textit{Conventional VQ}. This example demonstrates
the potential benefit of dynamic quantization according to the interference
topology profile. In the following, we shall elaborate the details
of the proposed scheme that can adapt to the general $\mathcal{ITP}$
given in Def. 1.

\subsection{\label{sub:Spatial-Correlation-Adaptation}Spatial Codebook Design}

In this section, we shall propose a novel spatial codebook design
to capture the asymmetric interference topology profile of MIMO interference
channel defined in Def. 1. From the toy example in the motivation
part, we see that better system performance can be achieved by deploying
spatial codebook given by transforming a base codebook with the corresponding
spatial correlation matrices. Based on this intuition, we illustrate
below how these spatial codebooks $\{\mathcal{C}_{ji}\}$ are designed.

\textit{Algorithm 2 (Spatial Codebook and Quantization Criterion)}:
Each codebook $\mathcal{C}_{ji}=\{\mathbf{W}_{ji}^{1},\cdots\mathbf{W}_{ji}^{N_{ji}}\}$,
$\mathbf{W}_{ji}^{l}\in\mathbb{C}^{N_{r}\times N_{t}}$ ($N_{ji}=2^{B_{ji}}$)
is designed by transforming a base codebook (the base codebooks can
be obtained by using the quantization cell approximation model in
\cite{yoo2007multi}) $\mathcal{C}_{ji}^{0}=\{\mathbf{S}_{ji}^{1}\cdots\mathbf{S}_{ji}^{N_{ji}}\}$,
where $\mathbf{S}_{ji}^{l}\in\mathbb{C}^{N_{r}\times N_{t}}$ with
the spatial correlation matrices $\mathbf{\Phi}_{ji}^{r}$, $\mathbf{\Phi}_{ji}^{t}$,
i.e.
\begin{equation}
\mathbf{W}_{ji}^{l}=\frac{(\mathbf{\Phi}_{ji}^{r})^{\nicefrac{1}{2}}\mathbf{S}_{ji}^{l}(\mathbf{\Phi}_{ji}^{t})^{\nicefrac{1}{2}}}{||(\mathbf{\Phi}_{ji}^{r})^{\nicefrac{1}{2}}\mathbf{S}_{ji}^{l}(\mathbf{\Phi}_{ji}^{t})^{\nicefrac{1}{2}}||}.\label{eq:codeword_mapping}
\end{equation}
With the input matrix $\mathbf{H}_{ji}$, the selected codeword $\mathbf{\hat{H}}_{ji}$
in codebook $\mathcal{C}_{ji}$ is given by 
\begin{equation}
\mathbf{\hat{H}}_{ji}=\textrm{arg}\max_{\mathbf{W}_{ji}^{l}\in\mathcal{C}_{ji}}||\textrm{vec}(\mathbf{H}_{ji})^{H}\textrm{vec}(\mathbf{W}_{ji}^{l})||.\label{eq:expression_for_E}
\end{equation}
 \hfill \IEEEQED
\begin{remrk}
[How the Spatial Codebooks adapts to the \textit{$\mathcal{ITP}$}]We
transform a base codebook $\mathcal{C}_{ji}^{0}$ with the spatial
correlation matrices $\mathbf{\Phi}_{ji}^{r},\,\mathbf{\Phi}_{ji}^{t}$
to obtain the spatial codebook $\mathcal{C}_{ji}$. Therefore, the
spatial distribution of the CSI $\mathbf{H}_{ji}$ matches the codewords
in the spatial codebook $\mathcal{C}_{ji}$ for all $j,\; i$, and
thus less average quantization distortion will be induced. Furthermore,
the quantization resolution of the spatial codebooks $\{\mathcal{C}_{ji}\}$
is also adaptive%
\footnote{This enables us to perform dynamic quantization among the cross-links,
which is discussed in detail in Section III-C.%
} to the heterogeneity of $\mathcal{ITP}$ among the cross-links to
further enhance the feedback efficiency. 
\end{remrk}

\begin{remrk}
[Linear Complexity of Spatial Codebook Design]In the above design,
the fixed base codebooks $\{\mathcal{C}_{ji}^{0}\}$ are stored at
both the Tx and Rx side. Whenever the spatial correlations change,
the new codebooks can be found by transforming the base codebooks
using the new spatial correlation matrices (\ref{eq:codeword_mapping}).
The overall complexity of the codebook design is thus $\mathcal{O}(N)$,
where $N$ is the number of codewords.
\end{remrk}

Before we can analyze the RINR, we have to quantify the quantization
distortion in terms of the bit allocation $\{B_{ji}\}$ as well as
the $\mathcal{ITP}$ parameters. The relationship between the actual
CSI $\mathbf{H}_{ji}$ and the quantized CSI $\mathbf{\hat{H}}_{ji}$
is given by: 
\begin{equation}
\mathbf{H}_{ji}=\alpha_{ji}\mathbf{\hat{H}}_{ji}+\Delta\mathbf{H}_{ji}\label{eq:csi_quantization_distortion}
\end{equation}
where $\alpha_{ji}$ is some unknown complex scalar and $\textrm{vec}(\Delta\mathbf{H}_{ji})$
is the \textit{quantization distortion} distributed in the orthogonal
complement space of $\textrm{vec}(\mathbf{\hat{H}}_{ji})$. The following
lemma gives an upper bound on the average quantization distortion.
\begin{lemma}
[Average Quantization Distortion]\textit{\label{lem:(CSI-Quantization-Distortions):}}Denote
$D_{ji}^{avg}=\mathbb{E}\{||\Delta\mathbf{H}_{ji}||^{2}\}$ as the
average quantization distortion. Under high-resolution assumption
(i.e., $B_{ji}$ is sufficiently large.), the average quantization
distortion $D_{ji}^{avg}$ is upper bounded by
\end{lemma}
\begin{equation}
D_{ji}^{avg}\leq D_{ji}^{upp}=\beta_{ji}\cdot2^{-\frac{B_{ji}}{M_{ji}^{r}M_{ji}^{t}-1}}\label{eq:UB_CSI}
\end{equation}
where the distortion coefficient $\beta_{ji}$ is given by\begin{small}
\begin{equation}
\beta_{ji}=\frac{\left(\prod_{m,n}\lambda_{ji,m}\sigma_{ji,n}\right)^{k_{ji,1}}}{2M_{ji}^{r}M_{ji}^{t}}\cdot\mathbb{E}\left\{ \left(\frac{\sum_{m,n}y_{mn}}{\sum_{m,n}\lambda_{ji,m}\sigma_{ji,n}y_{mn}}\right)^{k_{ji,2}}\cdot\left(\sum_{m,n}\lambda_{ji,m}\sigma_{ji,n}\left(N_{r}N_{t}-\lambda_{ji,m}\sigma_{ji,n}\right)y_{mn}\right)\right\} \label{eq:beta_expression}
\end{equation}
\end{small}where $k_{ji,1}=\frac{1}{M_{ji}^{r}M_{ji}^{t}-1}$, $k_{ji,2}=\frac{2M_{ji}^{r}M_{ji}^{t}-1}{M_{ji}^{r}M_{ji}^{t}-1}$,
$m\in\{1,\cdots M_{ji}^{r}\}$, $n\in\{1,\cdots M_{ji}^{t}\}$, and
each of $\left\{ y_{mn}\right\} $ is i.i.d. chi-square distributed
with degree of freedom $2$.
\begin{proof}
Please See Appendix \ref{sub:Proof-for-Lemma-quantization distorion}.
\end{proof}
\begin{remrk}
[Comparison with i.i.d. Case]When $\mathbf{\Phi}_{ji}^{r}=\mathbf{\Phi}_{ji}^{t}=\mathbf{I}$,
then $\mathbf{H}_{ji}$ is i.i.d. complex Gaussian distributed. From
(\ref{eq:beta_expression}), we get $\beta_{ji}=N_{r}N_{t}-1$ and
the distortion bound in (\ref{eq:UB_CSI}) reduces to 
\begin{equation}
D_{ji}^{avg}\leq D_{ji}^{upp}=(N_{r}N_{t}-1)\cdot2^{-\frac{B_{ji}}{N_{r}N_{t}-1}}=\mathbb{E}\{||\mathbf{H}_{ji}||^{2}\}\cdot\frac{N_{r}N_{t}-1}{N_{r}N_{t}}\cdot2^{-\frac{B_{ji}}{N_{r}N_{t}-1}}\label{eq:iid_distortion_bound}
\end{equation}
which is consistent with the distortion value derived in \cite{yoo2007multi,roh2006transmit}.
\end{remrk}

Based on the above lemma about quantization distortions, we obtain
the following theorem which describes an upper bound on the average
RINR.
\begin{thm}
[Upper Bound of Average RINR]\label{theorem 1:(Residue-Interference-Upper}Denote
$I_{j}^{avg}=\mathbb{E}\{I_{j}\}$ as the average RINR at Rx $j$
(\ref{eq:Residue_interference_expression}), under high-resolution
assumption (i.e., $B_{ji}$ is sufficiently large for all $i$, $i\neq j$),
$I_{j}^{avg}$ is upper bounded by 
\begin{equation}
I_{j}^{avg}\leq I_{j}^{upp}=Pd\cdot\sum_{i,i\neq j}^{K}\left(\frac{\beta_{ji}l_{ji}}{M_{ji}^{r}M_{ji}^{t}-1}\right)\cdot2^{-\frac{B_{ji}}{M_{ji}^{r}M_{ji}^{t}-1}}\label{eq:interference_leakage_UB}
\end{equation}
where $\beta_{ji}$ is given in \textit{Lemma }\ref{lem:(CSI-Quantization-Distortions):}.\end{thm}
\begin{proof}
See Appendix \ref{sub:Proof-for-the-Interference-upper}.
\end{proof}

\subsection{\label{sub:Path-Loss-Adaptation}Dynamic Quantization via Bit Allocations}

Based on the spatial codebook $\{\mathcal{C}_{ji}\}$ designed in
the previous section, we further perform dynamic quantization via
bits allocations $\{B_{ji}^{*}\}$ in order to exploit the \textit{heterogeneity}
of path loss and spatial correlations among different links. Denote
the sum feedback bits for all the cross links as $B$, we formulate
the dynamic quantization as follows, which aims to minimize the sum
of the average RINR upper bounds at all Rxs. 
\begin{problem}
[Dynamic Quantization via Bit Allocation]\label{Problem-1-(Optimal-feedback-policy-1}
\begin{eqnarray}
\underset{\{B_{ji},i\neq j\}}{\min} &  & \sum_{j=1}^{K}I_{j}^{upp}\nonumber \\
\textrm{s.t.} &  & \sum_{i,j,i\neq j}^{K}B_{ji}\leq B\label{eq:problem_formulation}
\end{eqnarray}
where $I_{j}^{upp}$ is given in \textit{Theorem 1}.\hfill \IEEEQED\end{problem}
\begin{thm}
[Bit Allocation Solution] The optimal solution to Problem 1 is given
by
\end{thm}
\begin{equation}
B_{ji}^{*}=\left[(M_{ji}^{r}M_{ji}^{t}-1)\left(\log\left(\frac{\beta_{ji}l_{ji}}{(M_{ji}^{r}M_{ji}^{t}-1)^{2}}\right)+b\right)\right]^{+}\label{eq:optimal_bits_allocation}
\end{equation}
where $b$ satisfies $\sum_{i,j,i\neq j}B_{ji}^{*}=B$.
\begin{proof}
Please see Appendix \ref{sub:Proof-for-Theorem-2}.\end{proof}
\begin{remrk}
[How the Dynamic Quantization adapts to the $\mathcal{ITP}$] We
shall use two examples to illustrate how the dynamic bit allocation
(\ref{eq:optimal_bits_allocation}) exploits the heterogeneity of
the $\mathcal{ITP}$. Consider the case when $\mathbf{\Phi}_{ji}^{r}=\mathbf{\Phi}^{r}$,
$\mathbf{\Phi}_{ji}^{t}=\mathbf{\Phi}^{t}$ for all $j,\, i$, $i\neq j$
Then $M_{ji}^{t}=M^{t}$, $M_{ji}^{r}=M^{r}$, $\beta_{ji}=\beta$
for all $j,\, i$, $i\neq j$, and thus we have $B_{ji}^{*}=\left[(M^{r}M^{t}-1)\left(\log\left(\frac{\beta l_{ji}}{(M^{r}M^{t}-1)^{2}}\right)+b\right)\right]^{+}$
according to (\ref{eq:optimal_bits_allocation}). From this expression,
we see that in this case, links with smaller path loss (larger value
of $l_{ji}$) will be allocated more bits; Consider the case that
$l_{ji}=l$, $\mathbf{\Phi}_{ji}^{r}=\mathbf{I}$ for all $j,\, i$,
$i\neq j$, and that $B$ is large such that $b$ will dominate $\log\left(\frac{\beta_{ji}l_{ji}}{(M_{ji}^{r}M_{ji}^{t}-1)^{2}}\right)$
in (\ref{eq:optimal_bits_allocation}) for all $j,\, i$, $i\neq j$
. Thus we get $B_{ji}^{*}\approx(M_{ji}^{r}M_{ji}^{t}-1)b=(N_{r}M_{ji}^{t}-1)b$.
From this expression, we see that in this case, links with smaller
$M_{ji}^{t}$, which corresponding to larger transmit spatial correlation%
\footnote{Smaller $M_{ji}^{t}$ ($1\leq M_{ji}^{t}\leq N_{t}$) means that the
channel matrix $\mathbf{H}_{ji}$ has smaller number of transmit directions,
which corresponds to larger transmit spatial correlations.%
}, will be allocated less feedback bits. With these adaptive allocations,
we can achieve less aggregate distortion and thus achieve less residual
interference after IA suppression. From these examples, we see that
the proposed dynamic quantization exploits the heterogeneity of $\mathcal{ITP}$
and hence the feedback efficiency is enhanced.
\end{remrk}

\section{Performance Analysis}

In this section, we analyze the network throughput of IA under limited
feedback for the $K$-user MIMO interference networks. We first derive
a network throughput lower bound (LB) for given average RINR at each
Rx. Combining this result with the upper bound of the average RINR
in \textit{Theorem }\ref{theorem 1:(Residue-Interference-Upper},
we obtain the network throughput LB under the proposed feedback design
$\mathcal{\mathcal{P}}=\{\mathcal{C}_{ji},\, B_{ji}^{*}\}$, and express
it in terms of the number of feedback bits $\{B_{ji}^{*}\}$, the
transmit SNR $P$ and the $\mathcal{ITP}$ parameters. Finally, we
show that when the number of feedback bits scales with SNR as $C_{s}\cdot\log\textrm{SNR}$,
the sum degrees of freedom of the network are preserved. Moreover,
the value of scaling coefficient $C_{s}$ can be significantly reduced
in networks with asymmetric interference topology.

We shall first impose the following assumption on the statistics of
the direct channels.
\begin{assumption}
[Direct Channel Statistics]\label{Direct-Channel-StatisticsAssume}Assume
that all the direct channels statistics are as follows: $\mathbf{\Phi}_{jj}^{r}=\mathbf{\Phi}_{jj}^{t}=\mathbf{I},\; l_{jj}=1,\;\forall j\in\{1,\cdots K\}$,
which corresponds to the i.i.d. Rayleigh fading model.\hfill \IEEEQED
\end{assumption}

Note that the IA scheme \cite{gomadam2011distributed} is only related
with the cross links. Hence, we give a simple channel model for the
direct links and focus on analyzing the limited feedback scheme for
the cross links to obtain elegant insights. Consider a joint decoding
strategy for the desired signal streams and denote $\{(\mathbf{U}_{j},\mathbf{V}_{j})\}$
as the perfect CSIT IA transceiver. Then the network throughput under
perfect CSIT can be expressed as \cite{gomadam2011distributed},

\begin{equation}
R_{per}=\sum_{j=1}^{K}\mathbb{E}\left\{ \log\textrm{det}\left(\mathbf{I}+\frac{P}{d}(\mathbf{U}_{j}^{H}\mathbf{H}_{jj}\mathbf{V}_{j})(\mathbf{U}_{j}^{H}\mathbf{H}_{jj}\mathbf{V}_{j})^{H}\right)\right\} .\label{eq:throughput_per_1-1}
\end{equation}

Following the above definition and treat residual interference as
noise, we define the network throughput under limited feedback as\begin{small}
\begin{equation}
R_{lim}=\sum_{j=1}^{K}\mathbb{E}\left\{ \log\textrm{det}\left(\mathbf{I}+\frac{P}{d}(\mathbf{\hat{U}}_{j}^{H}\mathbf{H}_{jj}\mathbf{\hat{V}}_{j})(\mathbf{\hat{U}}_{j}^{H}\mathbf{H}_{jj}\mathbf{\hat{V}}_{j})^{H}\left(\mathbf{I}+\frac{P}{d}\sum_{i\neq j}^{K}l_{ji}(\mathbf{\hat{U}}_{j}^{H}\mathbf{H}_{ji}\mathbf{\hat{V}}_{i})(\mathbf{\hat{U}}_{j}^{H}\mathbf{H}_{ji}\mathbf{\hat{V}}_{i})^{H}\right)^{-1}\right)\right\} \label{eq:pratical_throughput_1-1}
\end{equation}
\vspace{-0.3in}
\end{small}

\hspace{-11bp}where $\{(\mathbf{\hat{U}}_{j},\mathbf{\hat{V}}_{j})\}$
are the practical IA transceivers (\ref{eq:IA_processing}) designed
with quantized CSI.
\begin{thm}
[Throughput under Perfect CSIT]\label{lem:(Perfect-CSIT-Throughput):}$R_{per}$
can be expressed as follows
\end{thm}
\begin{equation}
R_{per}=Kd\int_{0}^{\infty}\log\left(1+\frac{P}{d}\cdot v\right)\cdot f(v)\textrm{d}v\label{eq:throughput_per_2-1}
\end{equation}
where $f(v)$ is the marginal probability density function (p.d.f.)
of the unordered eigenvalues of the \textit{$(d\times d)$ central
Wishart} matrix with $d$ degrees of freedom and covariance matrix
$\mathbf{I}$ ($\mathbf{W}_{d}(\mathbf{I},\; d)$) \cite{tulino2004random}
(closed-form expression of $f(v)$ can be found on page 32, \cite{tulino2004random}). 
\begin{proof}
See Appendix \ref{sub:Proof-for-Lem-perfect-throughput}.
\end{proof}

Due to the limited feedback, the network throughput $R_{lim}$ is
always upper bounded by $R_{per}$, i.e., $R_{lim}\leq R_{per}$.
In the following, we derive a LB of $R_{lim}$ under the proposed
feedback scheme. By decoupling the signal terms and interference terms,
and by deriving the convex property of $R_{lim}$ with respect to
(w.r.t.) the eigenvalues of the interference covariance matrix (Please
refer to Appendix \ref{sub:Proof-for-Theorem-Practical-Throughput}
for details), we obtain the following LB on $R_{lim}$ using Jensen's
inequality. 
\begin{lemma}
[Throughput LB for Given Average RINR]\label{Throughput-Lower-Bound}Given
average RINR $\mathbb{E}\{I_{j}\}$ at Rx $j$, the network throughput
$R_{lim}$ in (\ref{eq:pratical_throughput_1-1}) is lower bounded
by
\begin{equation}
R_{lim}\geq\sum_{j=1}^{K}d\cdot\int_{0}^{+\infty}\log\left(1+\frac{1}{d}\mathbb{E}\{I_{j}\}+\frac{P}{d}\cdot v\right)f(v)\textrm{d}v-\sum_{j=1}^{K}d\cdot\log\left(1+\frac{1}{d}\mathbb{E}\{I_{j}\}\right)\label{eq:general_bound}
\end{equation}
where $f(v)$ is given in \textit{Theorem} \ref{lem:(Perfect-CSIT-Throughput):}.\end{lemma}
\begin{proof}
Please See Appendix \ref{sub:Proof-for-Theorem-Practical-Throughput}.
\end{proof}
\begin{remrk}
[Advantages of the LB in Lemma \ref{Throughput-Lower-Bound}]\emph{Lemma}
\ref{Throughput-Lower-Bound} gives an approach to bound the network
throughput in terms of the average RINR at each Rx. First, it is applicable
to general $d$ $(d\geq1)$ data stream cases. Note that many previous
works on limited feedback system\cite{jindal2006mimo,yoo2007multi,cho2011feedback}
focus on single stream case ($d=1$) due to the mathematical difficulty
to analyze matrix functions. However, we start with matrix analysis
and overcome these difficulties to get a comparatively more general
result. Second, the LB provided in Lemma \ref{Throughput-Lower-Bound}
is tighter than the conventional analysis result in \cite{cho2011feedback}
(The LB in \cite{cho2011feedback} is equivalent to the right side
of (\ref{eq:general_bound}) by setting the first $\mathbb{E}\{I_{j}\}=0$.).
The numerical comparison of the two bounds is also illustrated in
Section V.
\end{remrk}

By combining \textit{Lemma \ref{Throughput-Lower-Bound}} with the
upper bound of the average RINR given in \textit{Theorem }\ref{theorem 1:(Residue-Interference-Upper},
we obtain the following throughput bound in terms of the number of
feedback bits $\{B_{ji}^{*}\}$, the transmit SNR $P$ and the $\mathcal{ITP}$
parameters.
\begin{thm}
[Throughput LB under Proposed Feedback Scheme]\label{Throughput-Bounds-under-PFS}Under
the proposed dynamic feedback scheme $\mathcal{\mathcal{P}}=\{\mathcal{C}_{ji},\, B_{ji}^{*}\}$,
$R_{lim}$ is lower bounded by
\begin{eqnarray}
 & R_{lim}\geq R_{low}=\sum_{j=1}^{K}d\cdot\int_{0}^{+\infty}\log\left(1+P\cdot\sum_{i\neq j}^{K}\left(\frac{\beta_{ji}l_{ji}}{M_{ji}^{r}M_{ji}^{t}-1}\right)\cdot2^{-\frac{B_{ji}^{*}}{M_{ji}^{r}M_{ji}^{t}-1}}+\frac{P}{d}\cdot v\right)f(v)\cdot\textrm{d}v\nonumber \\
 & -\sum_{j=1}^{K}d\cdot\log\left(1+P\cdot\sum_{i\neq j}^{K}\left(\frac{\beta_{ji}l_{ji}}{M_{ji}^{r}M_{ji}^{t}-1}\right)\cdot2^{-\frac{B_{ji}^{*}}{M_{ji}^{r}M_{ji}^{t}-1}}\right)\mathbf{}\label{eq:throughput_under_dynamic}
\end{eqnarray}
where $f(v)$ is given in \textit{Theorem} \ref{lem:(Perfect-CSIT-Throughput):},
$\beta_{ji}$ depends on $\mathcal{ITP}$ and is given in \textit{Lemma
}\ref{lem:(CSI-Quantization-Distortions):}.\end{thm}
\begin{proof}
By substituting the upper bound expression of $\mathbb{E}\{I_{j}\}$
in \textit{Theorem \ref{theorem 1:(Residue-Interference-Upper} }into
\textit{Lemma \ref{Throughput-Lower-Bound}, }we can get the desired
expression.\end{proof}
\begin{remrk}
[Interpretation of Theorem \ref{Throughput-Bounds-under-PFS}]For
IA under limited feedback in general MIMO interference networks, the
previous works \cite{krishnamachari2009interference}, \cite{thukral2009interference}
focus on analyzing the feedback bits scaling law. Thus it gives the
performance at extremely high feedback bits regime only and failed
to quantize the performance when we have finite feedback bits. However,
with the above result in \textit{Theorem }\ref{Throughput-Bounds-under-PFS},
we can quantize the throughput in terms of the transmit SNR $P$,
the number of feedback bits $B$ and $\mathcal{ITP}$ parameters.
For instance, consider a homogeneous i.i.d. fading case, i.e., $\mathbf{\Phi}_{ji}^{r}=\mathbf{\Phi}_{ji}^{t}=\mathbf{I},\; l_{ji}=1$
for all $i,$ $j$. Then the above term (\ref{eq:throughput_under_dynamic})
reduces to 
\[
R_{lim}\geq R_{low}=Kd\int_{0}^{+\infty}\log\left(1+\frac{P}{(1+\omega)d}\cdot v\right)f(v)\cdot\textrm{d}v=R_{per}\left(\frac{P}{1+\omega}\right)
\]
where $\omega=P(K-1)2^{-\frac{B}{K(K-1)(N_{r}N_{t}-1)}}$. This indicates
that $R_{lim}$ shall be no less than the throughput under perfect
CSIT with a power degradation ratio of $\frac{1}{1+\omega}$.
\end{remrk}

To obtain some simple insights on how the asymmetric $\mathcal{ITP}$
will affect our system performance under the proposed scheme, we give
the following corollary.
\begin{cor}
[Feedback Bits Scaling with SNR]\label{Feedback-Bits-Scaling}Denote
$\rho_{ji}=\frac{M_{ji}^{t}M_{ji}^{r}}{N_{t}N_{r}}$ ($0<\rho_{ji}\leq1)$.
When the number of sum feedback bits $B$ scales with $\log P$ as%
\footnote{Here the notation $I_{\{l_{ji}>0\}}$ denotes the indicator function.
In practice, we might never have exactly zero path gain. However,
when the path loss $l_{ji}$ is so large such that the interference
power is always below the noise floor within our SNR operation regime,
we can treat $l_{ji}=0$ and thus have $I_{\{l_{ji}>0\}}=0$.%
} 
\begin{equation}
B\geq\sum_{i,j,i\neq j}^{K}\left\{ I_{\{l_{ji}>0\}}\cdot(N_{r}N_{t}\rho_{ji}-1)\right\} \cdot\log P+C_{b}\label{eq:bits_scaling}
\end{equation}
where $C_{b}$ is some bounded constant independent of $P$, we have
that the sum DoFs of the network are preserved, i.e.,
\begin{equation}
\lim_{P\rightarrow\infty}\frac{R_{lim}}{\log P}=Kd.\label{eq:preserve_throughput_scaling}
\end{equation}
 \end{cor}
\begin{proof}
See Appendix \ref{sub:Proof-for-Corollary-Scaling_Power}.
\end{proof}
\begin{remrk}
[Interpretation of Corollary \ref{Feedback-Bits-Scaling}]We achieve
a similar result with \cite{krishnamachari2009interference,thukral2009interference}
on limited feedback analysis of IA on MIMO interference network, that
the sum DoF achievable by IA can be maintained when the number of
feedback bits scales on $\mathcal{O}(\log P)$ (\ref{eq:bits_scaling}).
Moreover, different from the i.i.d. channel fading model assumption
in these works, we consider a general asymmetric interference topology
and show how the asymmetric $\mathcal{ITP}$ can be exploited in the
proposed dynamic feedback scheme to reduce feedback bits. From \textit{Corollary
}\ref{Feedback-Bits-Scaling}, we see that in networks when the spatial
correlation matrices are not full rank, i.e., $M_{ji}^{r}=\textrm{rank}(\mathbf{\Phi}_{ji}^{r})<N_{r}$,
$M_{ji}^{t}=\textrm{rank}(\mathbf{\Phi}_{ji}^{t})<N_{t}$ $(\textrm{such that }\rho_{ji}<1)$
for some $j,\, i$, or in networks when the path loss is so large
such that $l_{ji}=0$ for some $j,\, i$, the scaling bits (\ref{eq:bits_scaling})
to maintain the sum DoFs of the system could be reduced.
\end{remrk}

\section{Numerical Results}

In this section, we verify the performance gain of the proposed scheme
through simulations. We shall first give the following random interference
topology model for the cross links. Note the direct links are still
assumed to have homogeneous path loss and i.i.d. fading (\emph{Assumption}
\ref{Direct-Channel-StatisticsAssume}).
\begin{definitn}
[Random Interference Topology Model]Assume the dynamics of the cross
link is contributed by both shadowing effect and transmit spatial
correlation. The shadowing effect is modeled by log-normal shadowing,
and the transmit spatial correlation is modeled using the Exponential
Correlation Model described in \cite{loyka2001channel}. Therefore,
in the $\mathcal{ITP}$, 
\[
\mathbf{\Phi}_{ji}^{r}=\mathbf{I},\;\mathbf{\Phi}_{ji}^{t}=\left[\begin{array}{cccc}
1 & \epsilon & \cdots & \epsilon^{N_{t}-1}\\
\epsilon^{*} & 1 &  & \epsilon^{N_{t}-2}\\
\vdots & \vdots & \ddots & \vdots\\
(\epsilon^{*})^{N_{t}-1} & (\epsilon^{*})^{N_{t}-2} & \cdots & 1
\end{array}\right],\; l_{ji}\sim\ln\mathcal{N}(u,{\normalcolor \delta^{2}),\;{\color{blue}{\normalcolor \forall i,j,j\neq i}}}
\]
where $\{l_{ji}\}$ for different cross-links are assumed to be i.i.d.,
$u$ is set to be $-\frac{1}{2}\delta^{2}$ to normalize the mean%
\footnote{The expectation of a log-normal distributed variable is $\mathbb{E}(\ln\mathcal{N}(u,\delta^{2}))=e^{u+\frac{1}{2}\delta^{2}}$.%
} of large fading parameters $l_{ji}$, $\forall j,i\neq j$ to be
1. \hfill \IEEEQED
\end{definitn}

Under the above model, the dynamics of the interference topology can
be expressed by two parameters, $|\epsilon|$ and $\delta^{2}$, which
stand for the dynamics of the spatial correlation and the dynamics
of shadowing effect respectively. Note that $|\epsilon|=0$ corresponds
to no correlation and $|\epsilon|=1$ corresponds to the strongest
correlation; while larger $\delta^{2}$ corresponds to larger dynamics
of the shadowing effect. In the following simulations, we compare
the performance of the proposed \textbf{dynamic feedback scheme} (DFS)
with the following baselines.
\begin{itemize}
\item \textbf{Conventional VQ} (CVQ): Each cross-link is allocated equal
feedback bits, and MIMO codebooks with symmetrically distributed codewords
are deployed to quantize the $\textrm{vec}(\mathbf{H})$. 
\item \textbf{Half Dynamic Scheme 1 }(HDS1): Deploy spatial codebooks (Section
III-B) but assign equal bits to all cross-links.
\item \textbf{Half Dynamic Scheme 2} (HDS2): Deploy symmetric codebooks
but dynamically allocate bits to the cross-links (Section III-C).
\item \textbf{Random Beamforming} (RB): Each Tx, Rx randomly choose a precoder
and decorrelator.
\end{itemize}

\subsection{Performance Comparison w.r.t. Amount of Feedback}

In Fig. \ref{fig:Throughput-Comparison-Versus-bits}, we consider
a $K=4$, $N_{t}=3$, $N_{r}=2$, $d=1$ MIMO interference network.
We vary the number of the feedback bits $B$ and compare the network
throughput of different schemes under the following parameter settings:
interference topology dynamics $(|\epsilon|,\,\delta^{2})=(0.7,\,3)$,
transmit SNR $10\log_{10}P=25\textrm{dB}$. It shows that DFS can
achieve higher throughput compared with the baselines, and larger
performance gain over CVQ is achieved in relatively higher feedback
bits regime. This shows that the proposed dynamic scheme can better
adapt to the interference topology and thus achieves less performance
degradation. On the other hand, we see that the proposed LB of DFS
(derived in this paper) can better bound the DFS than conventional
LB derived according to \cite{cho2011feedback}, especially in low
feedback bits regime.

\begin{figure}
\begin{centering}
\includegraphics[scale=0.65]{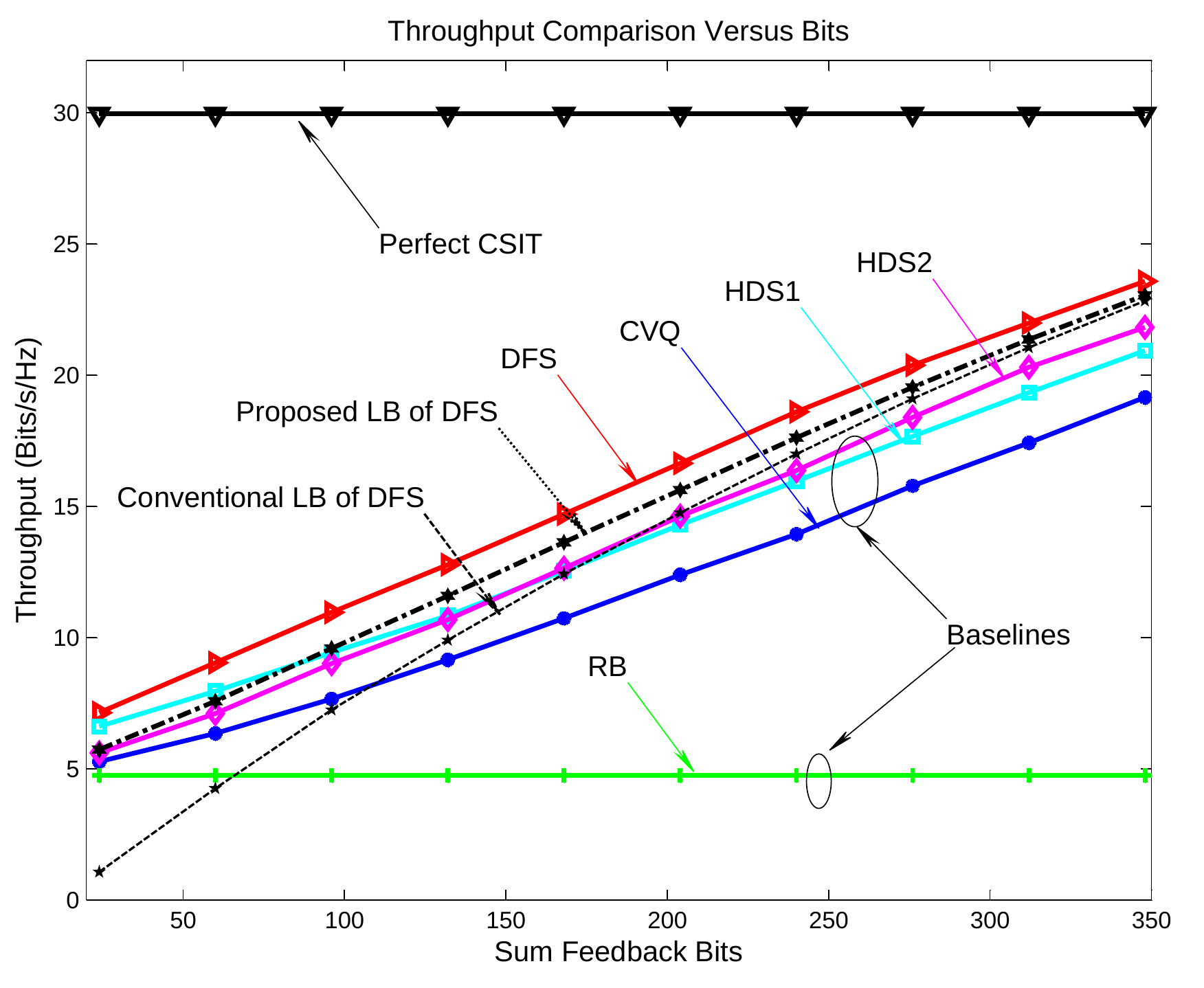}
\par\end{centering}

\caption{\label{fig:Throughput-Comparison-Versus-bits}Throughput comparison
versus sum feedback bits under$(|\epsilon|,\delta^{2})=(0.7,3)$ and
$10\log_{10}P=25\textrm{dB}$. }
\end{figure}

\subsection{Throughput Comparison w.r.t. Transmit SNR}

In Fig. \ref{fig:Throughput-power}, we consider a $K=4$, $N_{t}=3$,
$N_{r}=2$, $d=1$ MIMO interference network. We vary the transmit
SNR $P$ and compare the throughput of different schemes under the
following parameter settings: interference topology profile dynamics
$(|\epsilon|,\delta^{2})=(0.7,3)$, sum feedback bits $B=120$ and
$B=300$. The reason that we have two $B$ settings is to help illustrate
how the throughput goes with SNR in different feedback bits regimes.
It is shown that under both B settings, DFS can achieve a higher throughput
than the baselines. Moreover, larger performance gain is achieved
in the high SNR regime. 

\begin{figure}
\begin{centering}
\includegraphics[scale=0.65]{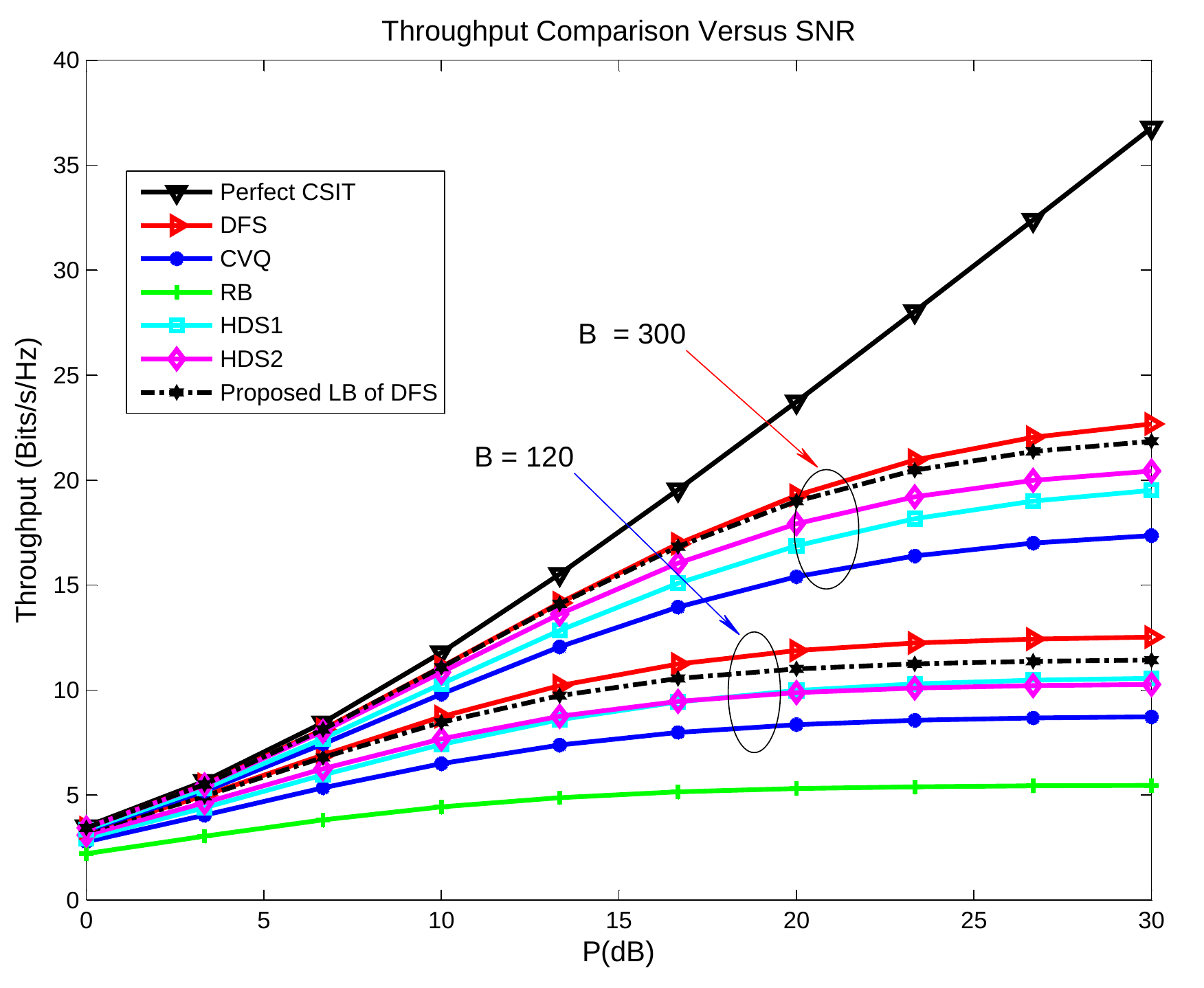}
\par\end{centering}

\caption{\label{fig:Throughput-power}Throughput comparison versus transmit
SNR under $(|\epsilon|,\delta^{2})=(0.7,3)$ and $B=120$, $300$.}
\end{figure}

In Fig \ref{fig:Throughput-scaling-versus-power}, we consider a $K=4$,
$N_{t}=3$, $N_{r}=2$, $d=1$ MIMO interference network. We vary
the transmit SNR $P$, scale the sum feedback bits with SNR as $B=K(K-1)(N_{r}N_{t}-1)\log P$
(see \textit{Corollary 1}) and show the throughput of different schemes
under interference topology dynamics $(|\epsilon|,\delta^{2})=(0.7,3)$.
From this figure, we can see that DFS achieves a larger throughput
compared with the baselines, which demonstrates its performance advantages.
Moreover, in the high SNR regime, we see that DFS, HDS1, HDS2 and
CVQ have the same slope as the perfect CSIT throughput. Therefore,
the sum DoFs of the network are maintained under this feedback bits
scaling condition. 

\begin{figure}
\begin{centering}
\includegraphics[scale=0.65]{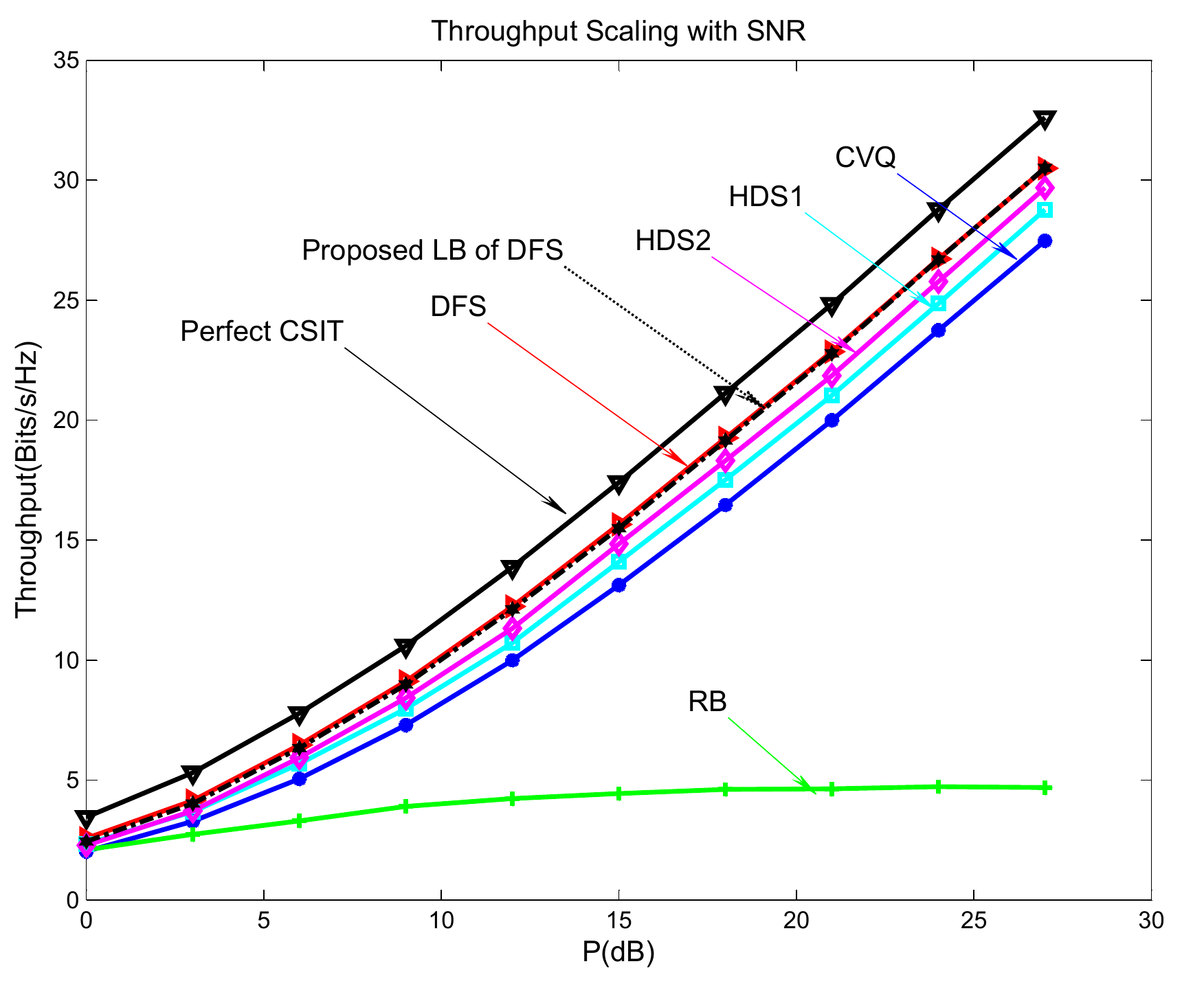}
\par\end{centering}

\caption{\label{fig:Throughput-scaling-versus-power}Throughput scaling with
transmit SNR under $(|\epsilon|,\delta^{2})=(0.7,3)$ and $B=K(K-1)(N_{r}N_{t}-1)\log P$.}
\end{figure}

\subsection{Throughput Comparison w.r.t. Interference Topology}

In Fig. \ref{fig:Throughput-comparison-versus-correlation} and Fig.
\ref{fig:Throughput-comparison-versus-pathloss}, we consider a $K=4$,
$N_{t}=6$, $N_{r}=4$, $d=2$ MIMO interference network. The reason
that we change to $d=2$ is to help verify that the proposed scheme
is also applicable to $d>1$ schemes. 

In Fig. \ref{fig:Throughput-comparison-versus-correlation}, we vary
the correlation coefficient $|\epsilon|$ and compare the network
throughput of different schemes under the following parameter settings:
shadowing dynamics $\delta^{2}=3$, transmit SNR $10\log_{10}P=25\textrm{dB}$
and sum feedback bits $B=828$. Note we choose a moderate number of
feedback bits to illustrate and compare the performance in residual-interference
limited region. From Fig. \ref{fig:Throughput-comparison-versus-correlation},
we observe that as $|\epsilon|$ goes higher, DFS and HDS1 achieve
larger performance gains over CVQ while HDS2 does not. This is because
DFS and HDS1 (but not HDS2) deploy the spatial codebook design (Section
III-B). From this fact, we conclude that the spatial codebook design
indeed captures the spatial correlations of the channel matrices and
thus improves the feedback efficiency. On the other hand, by comparing
HDS1 and HDS2, we see that under the proposed random interference
topology model, the spatial codebook design can contribute more to
the performance gain in the relatively higher spatial correlation
region.

\begin{figure}
\begin{centering}
\includegraphics[scale=0.65]{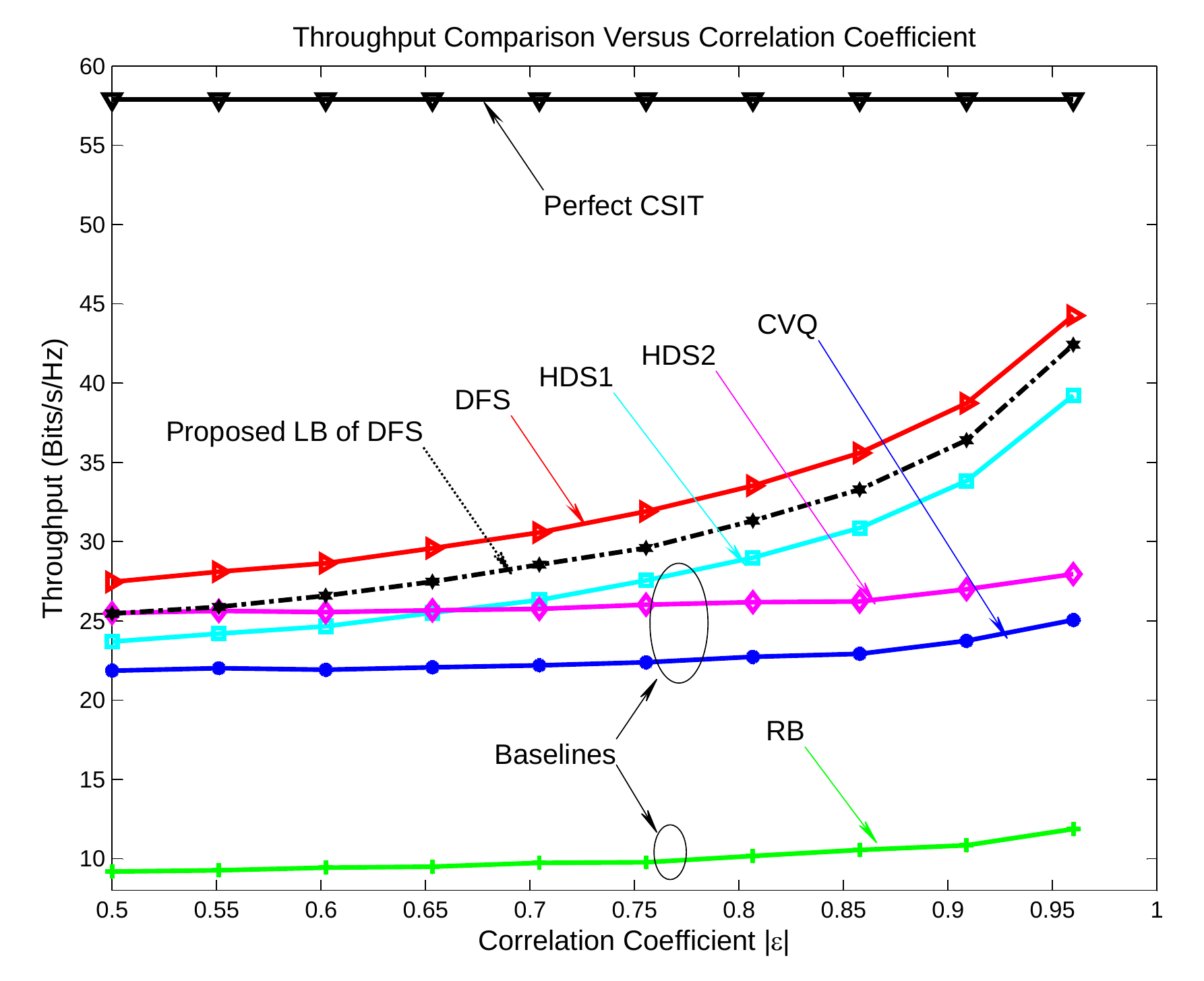}
\par\end{centering}

\caption{\label{fig:Throughput-comparison-versus-correlation}Throughput comparison
versus correlation coefficient $|\epsilon|$ under $\delta^{2}$ =
3, $10\log_{10}P$ = 25 dB and $B$ = 828.}
\end{figure}

In Fig. \ref{fig:Throughput-comparison-versus-pathloss}, we vary
the shadowing dynamics $\delta^{2}$ and compare the network throughput
of different schemes under the following parameter settings: correlation
coefficient $|\epsilon|=0.7$, transmit SNR $10\log_{10}P=25\textrm{dB}$
and sum feedback bits $B=828$. We see that as $\delta^{2}$ goes
higher, DFS and HDS2 achieve larger performance gains over CVQ while
HDS1 does not. This is because DFS and HDS2 (but not HDS1) deploy
the dynamic quantization via bit allocations. From this fact, we have
that the bit allocations indeed captures the shadowing dynamics and
thus improves the feedback efficiency. On the other hand, by comparing
HDS1 and HDS2, we see that under the proposed random interference
topology model, the bit allocations can contribute more to the performance
gain in the relatively higher shadowing dynamics region.

\begin{figure}
\begin{centering}
\includegraphics[scale=0.65]{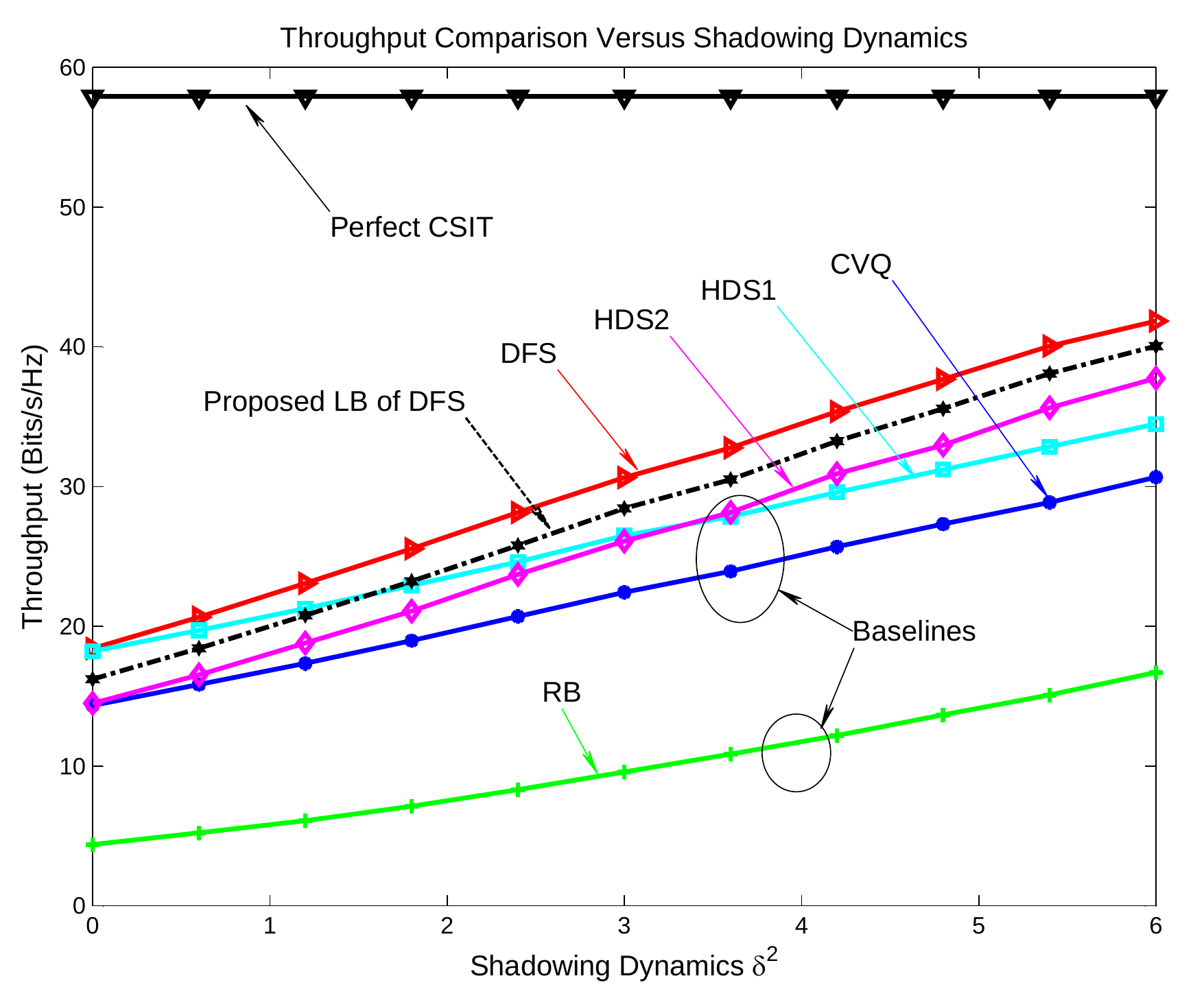}
\par\end{centering}

\caption{\label{fig:Throughput-comparison-versus-pathloss}Throughput comparison
versus shadowing dynamics $\delta^{2}$ under $|\epsilon|$ = 0.7,
$10\log_{10}P$ = 25 dB and $B$ = 828.}
\end{figure}

\section{Conclusions}

In this paper, we consider MIMO interference networks and investigate
the performance of IA under limited feedback. We consider a general
interference topology model which embraces various practical situations
such as spatial correlations and path loss effects. A novel spatial
codebook design with dynamic quantization is proposed to adapt to
the path loss and spatial correlations. We analyze the performance
bounds under the proposed dynamic feedback scheme, in terms of the
transmit SNR, feedback bits and the interference topology parameters.
Both analytical and simulation results show that the heterogeneity
of path loss and spatial correlations can be exploited in the proposed
scheme to enhance feedback efficiency.

\appendix

\subsection{\label{sub:Preliminaries_Transformed_Codebook}Preliminaries of Codebook
Design Heuristic: Transformed Codebook}

In \cite{love2006limited,xia2006design,zheng2008analysis}, the transformed
codebook design is proposed to improve the limited feedback performance
on correlated MISO channel. We now briefly illustrate the main ideas
of these works. Denote the MISO channel representation as 
\begin{equation}
y=\mathbf{h}^{H}\mathbf{f}\cdot x+n\label{eq:MISO_equation}
\end{equation}
where the correlated channel state is modeled as $\mathbf{h}=\mathbf{R}\mathbf{g}$,
$\mathbf{g}$ is i.i.d. complex Gaussian distributed with zero mean
and unit variance and $\mathbf{\Phi}_{T}=\mathbf{R}\mathbf{R}^{H}$is
the transmit correlation matrix known at the Tx side. The goal is
to design an efficient beamforming codebook (note the best beamforming
vector under perfect CSI should be $\mathbf{f}=\frac{\mathbf{h}}{||\mathbf{h}||}$)
so as to reduce the performance loss induced by limited feedback.
It is shown that the codebook given by: 
\begin{equation}
\left\{ \mathbf{f}_{i}\mid\mathbf{f}_{i}=\frac{\mathbf{R}\mathbf{c}_{i}}{||\mathbf{R}\mathbf{c}_{i}||},\;\mathbf{c}_{i}\in\mathcal{C}_{0}\right\} \label{eq:transformed_codebook}
\end{equation}
where $\mathcal{C}_{0}$ is a base codebook with codewords symmetrically
distributed in the Grassmannian subspace, can adapt to the channel
correlation and achieve near optimal performance. 

In \cite{zheng2008analysis,zheng2007analysis}, an upper bound of
the performance loss by using this transformed codebook is derived
via high-resolution asymptotic analysis. It is shown that the asymptotic
distortion of a finite rate feedback system is given by
\begin{equation}
D=\mathbb{E}\{D_{Q}(\mathbf{v},\mathbf{\hat{v}})\}=2^{-\frac{2B}{k_{q}}}\int_{\mathbb{Z}}\int_{\mathbb{Q}}m(\mathbf{v};\mathbf{z};\mathbb{E}_{\mathbf{z}}(\mathbf{v}))p(\mathbf{v},\mathbf{z})\lambda(\mathbf{v})^{-\frac{2}{k_{q}}}d\mathbf{v}d\mathbf{z}\label{eq:quantization_distortion_expression}
\end{equation}
where $D_{Q}$ is the user defined distortion function, $B$ the quantization
bits, $m(\mathbf{v};\mathbf{z};\mathbb{E}_{\mathbf{z}}(\mathbf{v}))$
the normalized inertial profile, $p(\mathbf{v},\mathbf{z})$ is the
probability density function at point $\mathbf{v}$ with side information
$\mathbf{z}$, and $\lambda(\mathbf{v})$ is the codeword point density.
Please refer to works \cite{zheng2008analysis,zheng2007analysis}
for the specific details.

\subsection{\label{sub:Proof-for-Lemma-quantization distorion}Proof for Lemma
\ref{lem:(CSI-Quantization-Distortions):} (CSI Quantization Distortion)}

The subscript $(i,j)$ is omitted for notation convenience in the
following derivations. Denote $\mathbf{\Phi}^{r}=\mathbf{F}^{r}\mathbf{\Lambda}^{r}(\mathbf{F}^{r})^{H}$,
$\mathbf{\Phi}^{t}=\mathbf{F}^{t}\mathbf{\Lambda}^{t}(\mathbf{F}^{t})^{H}$
as the eigenvalue decomposition, we have 
\begin{equation}
(\mathbf{\Phi}^{r})^{\nicefrac{1}{2}}=\mathbf{F}^{r}(\mathbf{\Lambda}^{r})^{\nicefrac{1}{2}}(\mathbf{F}^{r})^{H},\;(\mathbf{\Phi}^{t})^{\nicefrac{1}{2}}=\mathbf{F}^{t}(\mathbf{\Lambda}^{t})^{\nicefrac{1}{2}}(\mathbf{F}^{t})^{H}.\label{eq:csi_2}
\end{equation}
Define
\begin{equation}
\mathbf{E}=(\mathbf{F}^{r})^{H}\mathbf{H}\mathbf{F}^{t},\quad\mathbf{\hat{E}}=(\mathbf{F}^{r})^{H}\mathbf{\hat{H}}\mathbf{F}^{t},\quad\Delta\mathbf{E}=(\mathbf{F}^{r})^{H}\Delta\mathbf{H}\mathbf{F}^{t}.\label{eq:definition_of_e}
\end{equation}
Therefore, we have $\mathbf{E}=\alpha\mathbf{\hat{E}}+\Delta\mathbf{E}$
according to (\ref{eq:csi_quantization_distortion}). We shall prove
the following two lemmas first.
\begin{lemma}
\label{quantization-e}$\textrm{vec}(\Delta\mathbf{E})$ is distributed
in the orthogonal complement space of $\textrm{vec}(\mathbf{\hat{E}})$
and $\mathbb{E}\{||\Delta\mathbf{E}||^{2}\}=\mathbb{E}\{||\Delta\mathbf{H}||^{2}\}$.\end{lemma}
\begin{proof}
From (\ref{eq:csi_quantization_distortion}) and (\ref{eq:definition_of_e}),
we have: 
\begin{eqnarray*}
\textrm{vec}(\Delta\mathbf{H})^{H}\cdot\textrm{vec}(\hat{\mathbf{H}}) & = & \textrm{vec}(\Delta\mathbf{E})^{H}\cdot{\color{blue}{\normalcolor (({\normalcolor \mathbf{F}^{t}})^{*}}}\otimes\mathbf{F}^{r})^{H}\cdot{\normalcolor {\normalcolor {\color{blue}{\normalcolor (}{\normalcolor (\mathbf{F}^{t}}{\normalcolor )}^{{\normalcolor *}}}}}\otimes\mathbf{F}^{r})\cdot\textrm{vec}(\hat{\mathbf{E}})\\
 & \overset{(q_{0})}{=} & \textrm{vec}(\Delta\mathbf{E})^{H}\cdot\textrm{vec}(\hat{\mathbf{E}})=0
\end{eqnarray*}
where ($q_{0}$) comes from the fact that $(\mathbf{F}^{t})^{*}\otimes\mathbf{F}^{r}$
is a unitary matrix. Therefore, $\textrm{vec}(\Delta\mathbf{E})$
is distributed in the orthogonal complement space of $\textrm{vec}(\hat{\mathbf{E}})$.
The formula $\mathbb{E}\{||\Delta\mathbf{H}||^{2}\}=\mathbb{E}\{||\Delta\mathbf{E}||^{2}\}$
directly comes from the unitary invariance property of Frobenius norm
\cite{bernstein2011matrix}.\end{proof}
\begin{lemma}
\label{As-quantization-bits-expression}As quantization bits $B\rightarrow\infty$,
the asymptotic average distortion $\mathbb{E}\{||\Delta\mathbf{E}||^{2}\}$
is upper bounded by
\[
\mathbb{E}\{||\Delta\mathbf{E}||^{2}\}\leq\beta\cdot2^{-\frac{B}{M^{r}M^{t}-1}}
\]
where $\beta$ is a constant that depends on $\mathbf{\Phi}^{r}$,
$\mathbf{\Phi}^{t}$, $M^{r}=\textrm{rank}(\mathbf{\Phi}^{r})$, $M^{t}=\textrm{rank}(\mathbf{\Phi}^{r})$.\end{lemma}
\begin{proof}
Based on (\ref{eq:expression_for_E}), we have $\mathbf{\hat{H}}=\mathbf{W}^{l_{0}}$,
where $l_{0}$ is given by: 
\begin{eqnarray*}
l_{0} & = & \textrm{arg}\max_{1\leq l\leq2^{B}}||\textrm{vec}(\mathbf{H})^{H}\textrm{vec}(\mathbf{W}^{l})|=\textrm{arg}\max_{1\leq l\leq2^{B}}||\textrm{vec}(\mathbf{F}^{r}\mathbf{E}(\mathbf{F}^{t})^{H})^{H}\textrm{vec}(\mathbf{W}^{l})||\\
 & = & \textrm{arg}\max_{1\leq l\leq2^{B}}||\textrm{vec}(\mathbf{E})^{H}\cdot\textrm{vec}((\mathbf{F}^{r})^{H}\mathbf{W}^{l}(\mathbf{F}^{t}))||\\
 & \overset{(h_{1})}{=} & \textrm{arg}\max_{1\leq l\leq2^{B}}||\textrm{vec}(\mathbf{E})^{H}\cdot\textrm{vec}\left(\frac{(\mathbf{\Lambda}^{r})^{\nicefrac{1}{2}}(\mathbf{F}^{r})^{H}\mathbf{S}^{l}\mathbf{F}^{t}(\mathbf{\Lambda}^{t})^{\nicefrac{1}{2}}}{||(\mathbf{\Lambda}^{r})^{\nicefrac{1}{2}}(\mathbf{F}^{r})^{H}\mathbf{S}^{l}\mathbf{F}^{t}(\mathbf{\Lambda}^{t})^{\nicefrac{1}{2}}||}\right)||\\
 & = & \textrm{arg}\max_{1\leq l\leq2^{B}}||\textrm{vec}(\mathbf{E})^{H}\frac{(\mathbf{\Lambda}^{t}\otimes\mathbf{\Lambda}^{r})^{\nicefrac{1}{2}}\textrm{vec}\left((\mathbf{F}^{r})^{H}\mathbf{S}^{l}\mathbf{F}^{t}\right)}{||(\mathbf{\Lambda}^{t}\otimes\mathbf{\Lambda}^{r})^{\nicefrac{1}{2}}\textrm{vec}\left((\mathbf{F}^{r})^{H}\mathbf{S}^{l}\mathbf{F}^{t}\right)||}||,
\end{eqnarray*}
where in $(h_{1})$, $\mathbf{W}^{l}=\frac{\mathbf{F}^{r}(\mathbf{\Lambda}^{r})^{\nicefrac{1}{2}}(\mathbf{F}^{r})^{H}\mathbf{S}^{l}\mathbf{F}^{t}(\mathbf{\Lambda}^{t})^{\nicefrac{1}{2}}(\mathbf{F}^{t})^{H}}{||\mathbf{F}^{r}(\mathbf{\Lambda}^{r})^{\nicefrac{1}{2}}(\mathbf{F}^{r})^{H}\mathbf{S}^{l}\mathbf{F}^{t}(\mathbf{\Lambda}^{t})^{\nicefrac{1}{2}}(\mathbf{F}^{t})^{H}||}$
according to (\ref{eq:codeword_mapping}). Based on (\ref{eq:transformed_codebook}),
(\ref{eq:csi_2}), we can further obtain 
\begin{eqnarray*}
\textrm{vec}(\mathbf{\hat{E}}) & = & \textrm{vec}((\mathbf{F}^{r})^{H}\mathbf{\hat{H}}\mathbf{F}^{t})=\textrm{vec}((\mathbf{F}^{r})^{H}\mathbf{W}^{l_{0}}\mathbf{F}^{t})\\
 & = & \frac{(\mathbf{\Lambda}^{t}\otimes\mathbf{\Lambda}^{r})^{\nicefrac{1}{2}}\textrm{vec}\left((\mathbf{F}^{r})^{H}\mathbf{S}^{l_{0}}\mathbf{F}^{t}\right)}{||(\mathbf{\Lambda}^{t}\otimes\mathbf{\Lambda}^{r})^{\nicefrac{1}{2}}\textrm{vec}\left((\mathbf{F}^{r})^{H}\mathbf{S}^{l_{0}}\mathbf{F}^{t}\right)||}.
\end{eqnarray*}

Therefore, $\textrm{vec}(\mathbf{\hat{E}})$ can be regarded as the
selected codeword from the codebook $\mathcal{C}^{1}$ with input
vector $\textrm{vec}(\mathbf{E})$, i.e.,
\begin{equation}
\textrm{vec}(\mathbf{\hat{E}})=\textrm{arg}\max_{\mathbf{f}^{l}\in\mathcal{C}^{1}}||\textrm{vec}(\mathbf{E})^{H}\mathbf{f}_{l}||,\label{eq:E_expresssion_quantize}
\end{equation}
\begin{equation}
\mathcal{C}^{1}=\left\{ \mathbf{f}^{l}\mid\mathbf{f}^{l}=\frac{\mathbf{P}^{\frac{1}{2}}\textrm{vec}\left((\mathbf{F}^{r})^{H}\mathbf{S}^{l}\mathbf{F}^{t}\right)}{||\mathbf{P}^{\frac{1}{2}}\textrm{vec}\left((\mathbf{F}^{r})^{H}\mathbf{S}^{l}\mathbf{F}^{t}\right)||},\,\mathbf{S}^{l}\in\mathcal{C}^{0}\right\} ,\label{eq:new_codebook}
\end{equation}
where $\mathbf{P}=\mathbf{\Lambda}^{t}\otimes\mathbf{\Lambda}^{r}$. 

Before trying to calculate the quantization distortion, we shall first
eliminate the non effective dimensions \cite{zheng2008analysis}.
As $\textrm{vec}(\mathbf{E})=\mathbf{P}^{\frac{1}{2}}\textrm{vec}\left((\mathbf{F}^{r})^{H}\mathbf{H}^{w}\mathbf{F}^{t}\right)$
from (\ref{eq:definition_of_e}) where $\mathbf{\Lambda}^{r}=\textrm{diag}([\begin{array}{cccc}
\lambda_{1} & \cdots & \lambda_{M^{r}} & \mathbf{0}\end{array}])$, $\mathbf{\Lambda}^{t}=\textrm{diag}([\begin{array}{cccc}
\sigma & \cdots & \sigma_{M^{t}} & \mathbf{0}\end{array}])$, we have that there are $\left(N_{r}N_{t}-M^{r}M^{t}\right)$ entries
of $\textrm{vec}(\mathbf{E})$ always being null. Besides, the corresponding
entries in all codewords in codebook $\mathcal{C}^{1}$ are also zero
according to (\ref{eq:new_codebook}).

Denote the non-zero support (non-zero index set) for $\textrm{vec}(\mathbf{E})$
as $\mathbb{S}$ and the cardinality $|\mathbb{S}|=M^{t}M^{r}$. Denote
$\mathbf{g}(\mathbb{A})$ as the reduced vector formed by elements
of $\mathbf{g}$ whose index lies in $\mathbb{A}$. We get the following
reduced quantization model (which includes the source channel model,
the codebook, and the codeword selection criterion): 
\begin{equation}
\begin{cases}
\textrm{channel correlation model} & \mathbf{h}=\mathbf{T}^{\frac{1}{2}}\mathbf{x}_{m}\\
\textrm{transformed codebook} & \mathcal{C}^{\#}=\left\{ \mathbf{f}_{l}=\frac{\mathbf{T}^{\frac{1}{2}}\cdot\left(\textrm{vec}\left((\mathbf{F}^{r})^{H}\mathbf{S}^{l}\mathbf{F}^{t}\right)(\mathbb{S})\right)}{||\mathbf{T}^{\frac{1}{2}}\cdot\left(\textrm{vec}\left((\mathbf{F}^{r})^{H}\mathbf{S}^{l}\mathbf{F}^{t}\right)(\mathbb{S})\right)||}\mid\mathbf{S}^{l}\in\mathcal{C}^{0}\right\} \\
\textrm{codeword selection} & \mathbf{\hat{h}}=\textrm{arg}\max_{\mathbf{f}_{l}\in\mathcal{C}^{\#}}||\mathbf{h}^{H}\mathbf{f}_{l}||
\end{cases}\label{eq:mapped_MISO}
\end{equation}
where $\mathbf{h}=\textrm{vec}(\mathbf{E})(\mathbb{S})$, $\mathbf{\hat{h}}=\textrm{vec}(\hat{\mathbf{E}})(\mathbb{S})$,
$\mathbf{x}_{m}=\textrm{vec}\left((\mathbf{F}^{r})^{H}\mathbf{H}^{w}\mathbf{F}^{t}\right)(\mathbb{S})\in\mathbb{C}^{M^{r}M^{t}\times1}$,
and 
\begin{equation}
\mathbf{T}=\mathbf{P}(\mathbb{S},\,\mathbb{S})=\textrm{diag}([\begin{array}{ccccccc}
\lambda_{1}\sigma_{1}, & \lambda_{2}\sigma_{1} & \cdots & \lambda_{M^{r}}\sigma_{1}, & \lambda_{1}\sigma_{2} & \cdots & ,\lambda_{M^{r}}\sigma_{M^{t}}\end{array}]).\label{eq:Expression_T}
\end{equation}
As $\mathbf{H}^{w}$ is i.i.d. complex Gaussian distributed and $\mathbf{F}^{r}$,
$\mathbf{F}^{t}$ are unitary matrices, we have $(\mathbf{F}^{r})^{H}\mathbf{H}^{w}\mathbf{F}^{t}$
as well as $\mathbf{x}_{m}$ are i.i.d. complex Gaussian according
to the bi-unitarily invariant property of Gaussian random matrix \cite{tulino2004random}.
On the other hand, the codewords in $\{\textrm{vec}\left((\mathbf{F}^{r})^{H}\mathbf{S}^{l}\mathbf{F}^{t}\right)\mid\mathbf{S}^{l}\in\mathcal{C}^{0}\}$
are also isotropically distributed \cite{zheng2008analysis}. Therefore,
we can deploy the similar problem formulation described in \cite{zheng2008analysis}
and use the high-resolution asymptotic analysis to calculate the quantization
distortion. Define the distortion function to be
\begin{equation}
D_{Q}=||\textrm{vec}(\mathbf{E})||^{2}\cdot\left(1-|\left\langle \mathbf{v}_{e},\mathbf{\hat{v}}_{e}\right\rangle |^{2}\right)=||\mathbf{h}||^{2}\cdot\left(1-|\left\langle \mathbf{v}_{h},\mathbf{\hat{v}}_{h}\right\rangle |^{2}\right),\label{eq:proposed_distortion-1}
\end{equation}
where $\mathbf{v}_{e}=\frac{\textrm{vec}(\mathbf{E})}{||\textrm{vec}(\mathbf{E})||}$,
$\mathbf{v}_{h}=\frac{\mathbf{h}}{||\mathbf{h}||}$, and $\hat{\mathbf{v}}_{e}$,
$\mathbf{\hat{v}}_{h}$ are the corresponding quantized vectors.

We get the inertial profile $\tilde{m}_{tr-c}(\mathbf{v},\mathbf{h})$
\cite{zheng2008analysis} is upper bounded by
\begin{equation}
\tilde{m}_{tr-c}(\mathbf{v},\mathbf{h})\leq\frac{\gamma_{t}^{-1/(t-1)}||\mathbf{h}||^{2}(\mathbf{v}_{h}^{H}\mathbf{T}^{-1}\mathbf{v}_{h})}{t}\textrm{Tr}((\mathbf{I}-\mathbf{v}_{h}\mathbf{v}_{h}^{H})\mathbf{T}).\label{eq:inertial_us}
\end{equation}
where $t=M^{r}M^{t}$ (please refer to \cite{zheng2008analysis,zheng2007analysis}
for value of $\gamma_{t}$). The codeword density $\lambda(\mathbf{v})$
is 
\begin{equation}
\lambda(\mathbf{v})=\gamma_{t}^{-1}\cdot\textrm{det}(\mathbf{T}){}^{-1}\cdot(\mathbf{v}_{h}^{H}\mathbf{T}^{-1}\mathbf{v}_{h})^{-t}.\label{eq:codeword_density}
\end{equation}

By substituting (\ref{eq:inertial_us}), (\ref{eq:codeword_density})
into (\ref{eq:quantization_distortion_expression}), we get
\begin{equation}
\mathbb{E}(D_{Q})\leq\frac{\textrm{det}(\mathbf{T})^{\nicefrac{1}{t-1}}}{t}\mathbb{E}\left\{ \frac{\left(\mathbf{h}^{H}\mathbf{T}^{-1}\mathbf{h}\right)^{(\nicefrac{2t-1}{t-1})}\left(\textrm{Tr}(\mathbf{T})||\mathbf{h}||^{2}-\mathbf{h}^{H}\mathbf{T}\mathbf{h}\right)}{||\mathbf{h}||^{(\nicefrac{4t-2}{t-1})}}\right\} \cdot2^{-\frac{B}{t-1}}=\beta\cdot2^{-\frac{B}{M^{r}M^{t}-1}}.\label{eq:expectation_Dq}
\end{equation}

Substitute the expression of $\mathbf{T}$ (\ref{eq:Expression_T})
into the above formula, we can get the expression of $\beta$ as shown
in \textit{Lemma \ref{lem:(CSI-Quantization-Distortions):}}. For
the details of the omitted derivation, please refer to \cite{zheng2008analysis,zheng2007analysisDetailed}. 

We can express the relationship between the $\textrm{vec}(\mathbf{E})$
and $\textrm{vec}(\mathbf{\hat{E}})$ as: 
\[
\textrm{vec}(\mathbf{E})=||\mathbf{E}||\cdot|\left\langle \mathbf{v}_{e},\mathbf{\hat{v}}_{e}\right\rangle |e^{j\theta}\cdot\textrm{vec}(\mathbf{\hat{E}})+||\mathbf{E}||\cdot\left(1-|\left\langle \mathbf{v}_{e},\mathbf{\hat{v}}_{e}\right\rangle |^{2}\right)^{\nicefrac{1}{2}}\cdot\textrm{vec}(\mathbf{Z})
\]
where $\textrm{vec}(\mathbf{Z})$ is a random unit-norm vector distributed
in the orthogonal complement space of $\textrm{vec}(\mathbf{\hat{E}})$,
$\left\langle \cdot\right\rangle $ denotes the inner product operator
of two vectors. By setting $\alpha=||\mathbf{E}||\cdot|\left\langle \mathbf{v}_{e},\mathbf{\hat{v}}_{e}\right\rangle |e^{j\theta}$,
$\Delta\mathbf{E}=||\mathbf{E}||\cdot\left(1-|\left\langle \mathbf{v}_{e},\mathbf{\hat{v}}_{e}\right\rangle |^{2}\right)^{\nicefrac{1}{2}}\cdot\mathbf{Z}$,
we obtain 
\begin{equation}
\mathbb{E}\{||\Delta\mathbf{E}||^{2}\}=\mathbb{E}\{||\mathbf{E}||^{2}(1-|\left\langle \mathbf{v}_{e},\mathbf{\hat{v}}_{e}\right\rangle |^{2})\}=\mathbb{E}\{D_{Q}\}.\label{eq:delta_e_ex}
\end{equation}

Substitute the upper bound on $\mathbb{E}\{D_{Q}\}$ in (\ref{eq:expectation_Dq})
into (\ref{eq:delta_e_ex}), we get the lemma is proved.
\end{proof}

By combining \textit{Lemma \ref{quantization-e} }and \textit{Lemma
\ref{As-quantization-bits-expression}}, \textit{Lemma} \ref{lem:(CSI-Quantization-Distortions):}
is proved.

\subsection{\label{sub:Proof-for-the-Interference-upper}Proof for Theorem 1
(Upper Bound of average RINR)}

From (\ref{eq:csi_quantization_distortion}) and (\ref{eq:definition_of_e})
we have 
\[
\mathbf{H}_{ji}=\alpha_{ji}\mathbf{\hat{H}}_{ji}+\mathbf{F}_{ji}^{r}\Delta\mathbf{E}_{ji}(\mathbf{F}_{ji}^{t})^{H}.
\]

Based on (\ref{eq:Residue_interference_expression}) and the fact
that $\mathbf{\hat{U}}_{j}^{H}\mathbf{\hat{H}}_{ji}\mathbf{\hat{V}}_{i}=\mathbf{0}$
in (\ref{eq:IA_UV_property}), we have
\[
I_{j}=\dfrac{P}{d}\sum_{i,i\neq j}^{K}l_{ji}||\mathbf{\hat{U}}_{j}^{H}\mathbf{F}_{ji}^{r}\Delta\mathbf{E}_{ji}(\mathbf{F}_{ji}^{t})^{H}\mathbf{\hat{V}}_{i}||^{2}.
\]

We shall prove the following lemma first.
\begin{lemma}
\label{interference_leak-lemma}We have the following property $\forall j,\, i,\; i\neq j$,
\end{lemma}
\[
\mathbb{E}\{||\mathbf{\hat{U}}_{j}^{H}\mathbf{F}_{ji}^{r}\Delta\mathbf{E}_{ji}(\mathbf{F}_{ji}^{t})^{H}\mathbf{\hat{V}}_{i}||^{2}\}\leq\frac{d^{2}\beta_{ji}}{M_{ji}^{r}M_{ji}^{t}-1}\cdot2^{-\frac{B_{ji}}{M_{ji}^{r}M_{ji}^{t}-1}}.
\]

\begin{proof}
Denote $\mathbf{G}_{ji}^{l}=\mathbf{\hat{U}}_{j}^{H}\mathbf{F}_{ji}^{r}\in\mathbb{C}^{d\times N_{r}}$,
$\mathbf{G}_{ji}^{r}=(\mathbf{F}_{ji}^{t})^{H}\mathbf{\hat{V}}_{i}\in\mathbb{C}^{N_{t}\times d}$.
We get that $\mathbf{G}_{ji}^{l}$ has orthonormal rows and $\mathbf{G}_{ji}^{r}$
has orthonormal columns. Denote $\mathbf{G}_{ji}^{l}(m)$ as the $m$-th
row of $\mathbf{G}_{ji}^{l}$, $\mathbf{G}_{ji}^{r}(n)$ as the $n$-th
column of $\mathbf{G}_{ji}^{r}$, $1\leq m,n\leq d$. We get
\begin{eqnarray}
\mathbb{E}\{||\mathbf{\hat{U}}_{j}^{H}\mathbf{F}_{ji}^{r}\Delta\mathbf{E}_{ji}(\mathbf{F}_{ji}^{t})^{H}\mathbf{\hat{V}}_{i}||^{2}\} & = & \sum_{m,n}\mathbb{E}\left\{ ||\mathbf{G}_{ji}^{l}(m)\Delta\mathbf{E}_{ji}\mathbf{G}_{ji}^{r}(n)||^{2}\right\} \nonumber \\
 & = & \sum_{m,n}\mathbb{E}\left\{ ||\textrm{vec}\left(\mathbf{G}_{ji}^{l}(m)\Delta\mathbf{E}_{ji}\mathbf{G}_{ji}^{r}(n)\right)||^{2}\right\} \nonumber \\
 & = & \sum_{m,n}\mathbb{E}\left\{ ||\left(\mathbf{G}_{ji}^{r}(n)^{T}\otimes\mathbf{G}_{ji}^{l}(m)\right)\cdot\textrm{vec}\left(\Delta\mathbf{E}_{ji}\right)||^{2}\right\} \nonumber \\
 & = & \sum_{m,n}\mathbb{E}\left\{ ||\left(\mathbf{G}_{ji}^{r}(n)^{T}\otimes\mathbf{G}_{ji}^{l}(m)\right)(\mathbb{S}_{ji}^{T})\cdot\textrm{vec}\left(\Delta\mathbf{E}_{ji}(\mathbb{S}_{ji})\right)||^{2}\right\} \label{eq:sum_beta}
\end{eqnarray}
where $\mathbb{S}_{ji}$ denotes the non-zero support for the column
vector $\textrm{vec}(\mathbf{E}_{ji})$ (as explained in Lemma \ref{As-quantization-bits-expression}
in Appendix \ref{sub:Proof-for-Lemma-quantization distorion}), $\mathbb{S}_{ji}^{T}$
as the transpose of the support $\mathbb{S}_{ji}$ (note we use $\mathbb{S}_{ji}^{T}$
for $\mathbf{G}_{ji}^{r}(n)^{T}\otimes\mathbf{G}_{ji}^{l}(m)$ as
it is a row vector). 

We shall then illustrate two facts. First, we have $\mathbf{\hat{U}}_{j}^{H}\mathbf{F}_{ji}^{r}\mathbf{\hat{E}}_{ji}(\mathbf{F}_{ji}^{t})^{H}\mathbf{\hat{V}}_{i}=\mathbf{0}$
according to (\ref{eq:IA_UV_property}), and thus 
\[
\left(\mathbf{G}_{ji}^{r}(n)^{T}\otimes\mathbf{G}_{ji}^{l}(m)\right)(\mathbb{S}_{ji}^{T})\cdot\textrm{vec}\left(\mathbf{\hat{E}}_{ji}\right)(\mathbb{S}_{ji})=\left(\mathbf{G}_{ji}^{r}(n)^{T}\otimes\mathbf{G}_{ji}^{l}(m)\right)\cdot\textrm{vec}\left(\mathbf{\hat{E}}_{ji}\right)=0.
\]
Therefore, $\left(\left(\mathbf{G}_{ji}^{r}(n)^{T}\otimes\mathbf{G}_{ji}^{l}(m)\right)(\mathbb{S}_{ji}^{T})\right)^{H}$
lies in the $(M_{ji}^{r}M_{ji}^{t}-1)$ dimensional orthogonal complement
space of $\textrm{vec}\left(\mathbf{\hat{E}}_{ji}\right)(\mathbb{S}_{ji})$.
Second, under large quantization bits, the codeword density (\ref{eq:codeword_density})
near $\textrm{vec}(\mathbf{\hat{E}}_{ji})(\mathbb{S}_{ji})$ can be
approximated as uniform and thus$\textrm{vec}(\Delta\mathbf{E}_{ji})(\mathbb{S}_{ji})$
is approximately isotropically distributed in $(M_{ji}^{r}M_{ji}^{t}-1)$
dimensional orthogonal complement space of $\textrm{vec}(\mathbf{\hat{E}}_{ji})(\mathbb{S}_{ji})$.
Based on these two facts, we have

\begin{eqnarray}
 &  & \mathbb{E}\left\{ ||\left(\mathbf{G}_{ji}^{r}(n)^{T}\otimes\mathbf{G}_{ji}^{l}(m)\right)(\mathbb{S}_{ji}^{T})\cdot\textrm{vec}(\Delta\mathbf{E}_{ji})(\mathbb{S}_{ji})||^{2}\right\} \nonumber \\
 & = & \mathbb{E}\left\{ ||\left(\mathbf{G}_{ji}^{r}(n)^{T}\otimes\mathbf{G}_{ji}^{l}(m)\right)(\mathbb{S}_{ji}^{T})||^{2}||\textrm{vec}(\Delta\mathbf{E}_{ji})(\mathbb{S}_{ji})||^{2}\right\} \cdot\mathbb{E}\left(\textrm{beta}(1,M_{ji}^{r}M_{ji}^{t}-2)\right)\nonumber \\
 & \overset{(e)}{\leq} & \mathbb{E}\left\{ ||\textrm{vec}(\Delta\mathbf{E}_{ji})(\mathbb{S}_{ji})||^{2}\right\} \cdot\frac{1}{M_{ji}^{r}M_{ji}^{t}-1}\nonumber \\
 & = & \frac{\beta_{ji}}{M_{ji}^{r}M_{ji}^{t}-1}\cdot2^{-\frac{B_{ji}}{M_{ji}^{r}M_{ji}^{t}-1}}\label{eq:one_beta}
\end{eqnarray}
where $(e)$ comes from the fact that $||\mathbf{G}_{ji}^{r}(n)^{T}\otimes\mathbf{G}_{ji}^{l}(m)||=1$,
beta($\cdot$) denotes beta distribution \cite{yoo2007multi}. 

By combining (\ref{eq:sum_beta}) and (\ref{eq:one_beta}), \textit{Lemma
\ref{interference_leak-lemma} }is proved.
\end{proof}

Based on\textit{ Lemma \ref{interference_leak-lemma}, }we easily
get: 
\[
\mathbb{E}\{I_{j}\}\leq\dfrac{P}{d}\sum_{i,i\neq j}^{K}l_{ji}\mathbb{E}\{||\mathbf{\hat{U}}_{j}^{H}\mathbf{F}_{ji}^{r}\Delta\mathbf{E}_{ji}(\mathbf{F}_{ji}^{t})^{H}\mathbf{\hat{V}}_{i}||^{2}\}\leq Pd\sum_{i,i\neq j}^{K}\left(\frac{\beta_{ji}l_{ji}}{M_{ji}^{r}M_{ji}^{t}-1}\right)2^{-\frac{B_{ji}}{M_{ji}^{r}M_{ji}^{t}-1}}.
\]

\subsection{\label{sub:Proof-for-Theorem-2}Proof for Theorem 2 (Bits Allocation
Solution)}

Formulate the Lagrangian with multiplier $\gamma$, and set the derivative
w.r.t. $B_{ji}$ and $\gamma$ to zero 
\[
L=Pd\cdot\sum_{i.j,i\neq j}^{K}\left(\frac{\beta_{ji}l_{ji}}{M_{ji}^{r}M_{ji}^{t}-1}\right)\cdot2^{-\frac{B_{ji}}{M_{ji}^{r}M_{ji}^{t}-1}}+\gamma\left(\sum_{i,j,i\neq j}^{K}B_{ji}-B\right).
\]
\begin{equation}
\frac{\partial L}{\partial B_{ji}}=-Pd\cdot\frac{\beta_{ji}l_{ji}\ln2}{(M_{ji}^{r}M_{ji}^{t}-1)^{2}}\cdot2^{-\frac{B_{ji}}{M_{ji}^{r}M_{ji}^{t}-1}}+\gamma=0.\label{eq:kkt_1}
\end{equation}
\begin{equation}
\frac{\partial L}{\partial\gamma}=\sum_{i,j,i\neq j}^{K}B_{ji}-B=0.\label{eq:kkt_2}
\end{equation}

From (\ref{eq:kkt_1}), we get 
\[
B_{ji}=(M_{ji}^{r}M_{ji}^{t}-1)\left(\log\left(\frac{\beta_{ji}l_{ji}}{(M_{ji}^{r}M_{ji}^{t}-1)^{2}}\right)+b\right)
\]
where $b=\log\frac{Pd\cdot\ln2}{{\color{blue}\gamma}}$. Combine the
above expression with (\ref{eq:kkt_2}) as well as the condition that
$B_{ji}\geq0$, we get the desired solution.

\subsection{\label{sub:Proof-for-Lem-perfect-throughput}Proof for Theorem 3
(Throughput under Perfect CSIT )}

Given any $(\mathbf{U}_{j},\mathbf{V}_{j})$, we can construct unitary
matrices $\mathbf{\dot{U}}_{j}=\left[\begin{array}{cc}
\mathbf{U}_{j} & \mathbf{U}_{j}^{c}\end{array}\right]$, $\mathbf{\dot{V}}_{j}=\left[\begin{array}{cc}
\mathbf{V}_{j} & \mathbf{V}_{j}^{c}\end{array}\right]$. As $\{(\mathbf{U}_{j},\mathbf{V}_{j})\}$ are independent of the
direct channel states $\{\mathbf{H}_{jj}\}$, we have that $\{(\mathbf{\dot{U}}_{j},\mathbf{\dot{V}}_{j})\}$
are also independent of them. By combining this feature with the fact
that $\{\mathbf{H}_{jj}\}$ are i.i.d. complex Gaussian distributed,
we get $\mathbf{\dot{U}}_{j}^{H}\mathbf{H}_{jj}\mathbf{\dot{V}}_{j}$
is also i.i.d. complex Gaussian distributed according to the \textit{bi-unitarily
invariant property} of i.i.d. complex Gaussian matrix \cite{tulino2004random}.
Therefore, 
\begin{equation}
\mathbf{U}_{j}^{H}\mathbf{H}_{jj}\mathbf{V}_{j}=\mathbf{U}_{j}^{H}(\mathbf{\dot{U}}_{j}\mathbf{\dot{U}}_{j}^{H})\mathbf{H}_{jj}(\mathbf{\dot{V}}_{j}\mathbf{\dot{V}}_{j}^{H})\mathbf{V}_{j}=\left[\begin{array}{cc}
\underset{d\times d}{\mathbf{I}} & \mathbf{0}\end{array}\right](\mathbf{\dot{U}}_{j}^{H}\mathbf{H}_{jj}\mathbf{\dot{V}}_{j})\left[\begin{array}{c}
\mathbf{I}_{d\times d}\\
\mathbf{0}
\end{array}\right]=\underset{d\times d}{\tilde{\mathbf{H}}_{jj}}\label{eq:unitary_diagonal}
\end{equation}
where $\tilde{\mathbf{H}}_{jj}$ denotes the left upper $(d\times d)$
sub matrix of $(\mathbf{\dot{U}}_{j}^{H}\mathbf{H}_{jj}\mathbf{\dot{V}}_{j})$,
and is thus i.i.d. complex Gaussian distributed. Therefore, $\tilde{\mathbf{H}}_{jj}\tilde{\mathbf{H}}_{jj}^{H}$
is a central Wishart matrix with degree of freedom $d$ and covariance
matrix $I_{d}$ . We have,
\[
R_{per}=\sum_{j=1}^{K}\mathbb{E}\left[\log\textrm{det}\left(\mathbf{I}+\frac{P}{d}\tilde{\mathbf{H}}_{jj}\tilde{\mathbf{H}}_{jj}^{H}\right)\right]=Kd\int_{0}^{\infty}\log(1+\frac{P}{d}\cdot v)f(v)\textrm{d}v
\]
where $f(v)$ is the marginal probability density (p.d.f.) function
of the unordered eigenvalues of the\textit{ $(d\times d)$ central
Wishart} matrix with $d$ degrees of freedom and covariance matrix
$\mathbf{I}_{d}$ $\mathbf{W}_{d}(\mathbf{I}_{d},\; d)$) \cite{tulino2004random}
(closed-form expression of $f(v)$ can be found in page 32, \cite{tulino2004random}
).

\subsection{\label{sub:Proof-for-Theorem-Practical-Throughput}Proof for Lemma
2 (Throughput LB for Given RINR)}

For any given $(\mathbf{\hat{U}}_{i},\mathbf{\hat{V}}_{i})$, we can
construct unitary matrices $\mathbf{\bar{U}}_{j}=\left[\begin{array}{cc}
\mathbf{\hat{U}}_{j} & \mathbf{\hat{U}}_{j}^{c}\end{array}\right]$, $\mathbf{\bar{V}}_{j}=\left[\begin{array}{cc}
\mathbf{\hat{V}}_{j} & \mathbf{\hat{V}}_{j}^{c}\end{array}\right]$. As $\{(\mathbf{\bar{U}}_{j},\mathbf{\bar{V}}_{j})\}$ are independent
of the i.i.d. complex Gaussian matrix $\mathbf{H}_{jj}$, we have
$\mathbf{\bar{U}}_{j}^{H}\mathbf{H}_{jj}\mathbf{\bar{V}}_{j}$ is
also i.i.d. complex Gaussian distributed \cite{tulino2004random}.
Therefore, 
\[
\mathbf{\hat{U}}_{j}^{H}\mathbf{H}_{jj}\mathbf{\hat{V}}_{j}=\mathbf{\hat{U}}_{j}^{H}(\mathbf{\bar{U}}_{j}\mathbf{\bar{U}}_{j}^{H})\mathbf{H}_{jj}(\mathbf{\bar{V}}_{j}\mathbf{\bar{V}}_{j}^{H})\mathbf{\hat{V}}_{j}=\left[\begin{array}{cc}
\mathbf{I}_{d} & \mathbf{0}\end{array}\right](\mathbf{\bar{U}}_{j}^{H}\mathbf{H}_{jj}\mathbf{\bar{V}}_{j})\left[\begin{array}{c}
\mathbf{I}_{d}\\
\mathbf{0}
\end{array}\right].
\]
Hence,$\mathbf{\hat{U}}_{j}^{H}\mathbf{H}_{jj}\mathbf{\hat{V}}_{j}$
(the left upper $(d\times d)$ sub matrix of $\mathbf{\bar{U}}_{j}^{H}\mathbf{H}_{jj}\mathbf{\bar{V}}_{j}$)
is i.i.d. complex Gaussian and is statistically independent of $(\mathbf{\hat{U}}_{j},\mathbf{\hat{V}}_{j})$
according to the \textit{bi-unitarily invariant property} of i.i.d.
complex Gaussian matrix \cite{tulino2004random}. On the other hand,
as $\mathbf{H}_{jj}$ is independent of $\{\mathbf{H}_{ji},\, i\neq j\}$,
we get that $\hat{\mathbf{U}}_{j}^{H}\mathbf{H}_{jj}\mathbf{\hat{V}}_{j}$
and $\mathbf{\hat{U}}_{j}^{H}\mathbf{H}_{ji}\mathbf{\hat{V}}_{i}$
(for all $i\neq j$ ) are conditionally independent given $\{(\mathbf{\hat{U}}_{j},\mathbf{\hat{V}}_{j})\}$.
Combine this feature with the fact that $\mathbf{\hat{U}}_{j}^{H}\mathbf{H}_{jj}\mathbf{\hat{V}}_{j}$
is statistically independent of $\{(\mathbf{\hat{U}}_{j},\mathbf{\hat{V}}_{j})\}$,
we have that $\mathbf{\hat{U}}_{j}^{H}\mathbf{H}_{jj}\mathbf{\hat{V}}_{j}$
is independent of $\mathbf{\hat{U}}_{j}^{H}\mathbf{H}_{ji}\mathbf{\hat{V}}_{i}$
for all $i\neq j$. Hence, we get that the desired signal and interference
signal are decoupled. Denote 
\begin{equation}
\frac{P}{d}\sum_{i\neq j}^{K}l_{ji}(\mathbf{\hat{U}}_{j}^{H}\mathbf{H}_{ji}\mathbf{\hat{V}}_{i})(\mathbf{\hat{U}}_{j}^{H}\mathbf{H}_{ji}\mathbf{\hat{V}}_{i})^{H}=\mathbf{R}_{j}\mathbf{\Sigma}_{j}\mathbf{R}_{j}^{H}\label{eq:expression_diagonalization}
\end{equation}
as the eigenvalue decomposition, where $\mathbf{R}_{j}$ is a unitary
matrix and $\mathbf{\Sigma}_{j}$ is the diagonal matrix with real
positive eigenvalues. We have that the i.i.d. complex Gaussian matrix
$\mathbf{\hat{U}}_{j}^{H}\mathbf{H}_{jj}\mathbf{\hat{V}}_{j}$ is
independent of the unitary matrix $\mathbf{R}_{j}$. Therefore, we
have $\mathbf{H}_{j}^{q}=\mathbf{R}_{j}^{H}\mathbf{\hat{U}}_{j}^{H}\mathbf{H}_{jj}\mathbf{\hat{V}}_{j}$
is also i.i.d. complex Gaussian distributed statistically independent
of $\mathbf{R}_{j}^{H}$, and thus is also independent of $\mathbf{\Sigma}_{j}$.
Therefore, we can first take expectation w.r.t. $\mathbf{\Sigma}_{j}$,
and then w.r.t. $\mathbf{H}_{j}^{q}$ for $R_{lim}$ in (\ref{eq:pratical_throughput_1-1}),
i.e., 
\begin{equation}
R_{lim}=\sum_{j=1}^{K}\mathbb{E}_{\{\mathbf{H}_{j}^{q}\}}\left\{ \mathbb{E}_{\{\mathbf{\Sigma}_{j}\}}\left\{ \log\textrm{det}\left(\mathbf{I}+\frac{P}{d}\left(\mathbf{H}_{j}^{q}(\mathbf{H}_{j}^{q})^{H}\right)\left(\mathbf{I}+\mathbf{\Sigma}_{j}\right)^{-1}\right)\right\} \right\} .\label{eq:reduced_expression}
\end{equation}
To help prove the theorem, we shall first prove the following lemma.
\begin{lemma}
\label{-The-real-convexity} The function $g(\mathbf{X})=\log\textrm{det}\left(\mathbf{I}+\mathbf{A}\mathbf{X}^{-1}\right)$
is convex w.r.t. $\mathbf{X}$, where $\mathbf{A}\in\mathbb{C}^{d\times d}$
is a constant \textit{Hermitian} positive definite (PD) matrix and
$\mathbf{X}$ is defined on $\mathbb{D}_{d}=\left\{ \textrm{diag}([\begin{array}{ccc}
x_{1} & \cdots & x_{d}\end{array}])\mid x_{i}>0,\,\forall i\right\} $. \end{lemma}
\begin{proof}
$g(\mathbf{X})=\log\textrm{det}\left(\mathbf{X}+\mathbf{A}\right)-\log\textrm{det}\left(\mathbf{X}\right)$.
According to \cite{hjorungnes2007complex}, the second order differential
of $g(\mathbf{X})$ is 
\[
\textrm{d}^{2}g(\mathbf{X})=\frac{1}{\ln2}\cdot\textrm{d}\textrm{vec}(\mathbf{X})^{T}\cdot\mathcal{H}_{\mathbf{X},\mathbf{X}}g(\mathbf{X})\cdot\textrm{d}\textrm{vec}(\mathbf{X})
\]
where $\mathcal{H}_{\mathbf{X},\mathbf{X}}g(\mathbf{X})=-((\mathbf{X}+\mathbf{A})^{T})^{-1}\otimes(\mathbf{X}+\mathbf{A})^{-1}+(\mathbf{X}^{T})^{-1}\otimes\mathbf{X}^{-1}$.
Since $\mathbf{X}+\mathbf{A}\succeq\mathbf{X}$ and both $\mathbf{X}+\mathbf{A}$,$\mathbf{X}\succeq$0
(here $\mathbf{A}\succeq\mathbf{B}$ means that $\mathbf{A}-\mathbf{B}$
is PD), then $(\mathbf{X}+\mathbf{A})^{-1}\preceq\mathbf{X}^{-1}$,
$((\mathbf{X}+\mathbf{A})^{T})^{-1}\preceq(\mathbf{X}^{T})^{-1}$,
and it is easy to verify that $\mathcal{H}_{\mathbf{X},\mathbf{X}}g(\mathbf{X})\succeq0$
\cite{bernstein2011matrix}. Therefore, $g(\mathbf{X})$ is convex
w.r.t. $\mathbf{X}$. 
\end{proof}

With the convexity property in \textit{Lemma \ref{-The-real-convexity}}
and using the Jensen's Inequality on (\ref{eq:reduced_expression}),
we have 
\[
R_{lim}{\color{blue}{\normalcolor \geq}}\sum_{j=1}^{K}\mathbb{E}\left\{ \log\textrm{det}\left(\mathbf{I}+\frac{P}{d}\left(\mathbf{H}_{j}^{q}(\mathbf{H}_{j}^{q})^{H}\right)\left(\mathbf{I}+\mathbb{E}\left\{ \mathbf{\Sigma}_{j}\right\} \right)^{-1}\right)\right\} 
\]
and $\textrm{Tr}\left(\mathbb{E}\left\{ \mathbf{\Sigma}_{j}\right\} \right)=\mathbb{E}\left\{ \textrm{Tr}(\mathbf{\Sigma}_{j})\right\} =\mathbb{E}\left\{ \textrm{Tr}(\frac{P}{d}\sum_{i\neq j}^{K}l_{ji}(\mathbf{\hat{U}}_{j}^{H}\mathbf{H}_{ji}\mathbf{\hat{V}}_{i})(\mathbf{\hat{U}}_{j}^{H}\mathbf{H}_{ji}\mathbf{\hat{V}}_{i})^{H})\right\} =\mathbb{E}(I_{j})$. 

Denote $\mathbf{P}_{d}$ as a permutation matrix with dimension $d$,
and the set of all permutation matrices with dimension $d$ as $\mathbb{P}_{d}$.
Since $\mathbf{H}_{j}^{q}(\mathbf{H}_{j}^{q})^{H}$ is a \textit{central
Wishart matrix}, we have that

\begin{eqnarray*}
R_{lim} & \geq & \sum_{j=1}^{K}\mathbb{E}\left\{ \log\textrm{det}\left(\mathbf{I}+\frac{P}{d}\left(\mathbf{H}_{j}^{q}(\mathbf{H}_{j}^{q})^{H}\right)\left(\mathbf{I}+\mathbb{E}\left\{ \mathbf{\Sigma}_{j}\right\} \right)^{-1}\right)\right\} \\
 & = & \sum_{j=1}^{K}\mathbb{E}\left\{ \log\textrm{det}\left(\mathbf{I}+\frac{P}{d}\left(\mathbf{H}_{j}^{q}(\mathbf{H}_{j}^{q})^{H}\right)\left(\mathbf{I}+\mathbf{P}_{d}\cdot\mathbb{E}\left\{ \mathbf{\Sigma}_{j}\right\} \cdot\mathbf{P}_{d}\right)^{-1}\right)\right\} 
\end{eqnarray*}
for any $\mathbf{P}_{d}\in\mathbb{P}_{d}$. Further using Jensen's
inequality, we get

\begin{eqnarray*}
R_{lim} & \geq & \sum_{j=1}^{K}\mathbb{E}\left\{ \log\textrm{det}\left(\mathbf{I}+\frac{P}{d}\left(\mathbf{H}_{j}^{q}(\mathbf{H}_{j}^{q})^{H}\right)\left(\mathbf{I}+\frac{1}{d!}\sum_{\mathbf{P}_{d}\in\mathbb{P}_{d}}\mathbf{P}_{d}\cdot\mathbb{E}\left\{ \mathbf{\Sigma}_{j}\right\} \cdot\mathbf{P}_{d}^{T}\right)^{-1}\right)\right\} \\
 & = & \sum_{j=1}^{K}\mathbb{E}\left\{ \log\textrm{det}\left(\mathbf{I}+\frac{P}{d}\left(\mathbf{H}_{j}^{q}(\mathbf{H}_{j}^{q})^{H}\right)\left(\mathbf{I}+\frac{\mathbb{E}\{I_{j}\}}{d}\cdot\mathbf{I}\right)^{-1}\right)\right\} \\
 & \overset{(r)}{=} & \sum_{j=1}^{K}d\cdot\int_{0}^{+\infty}\log\left(1+\frac{1}{d}\mathbb{E}\{I_{j}\}+\frac{P}{d}\cdot v\right)f\textrm{(v)d}v-\sum_{j=1}^{K}d\cdot\log\left(1+\frac{1}{d}\mathbb{E}\{I_{j}\}\right)
\end{eqnarray*}
where in $(r)$, $\mathbf{H}_{j}^{q}(\mathbf{H}_{j}^{q})^{H}$ is
a \textit{central Wishart} matrix with $d$ degrees of freedom and
covariance matrix $\mathbf{I}_{d}$ ($\mathbf{W}_{d}(\mathbf{I}_{d},\; d)$)
), $f(v)$ is given in Theorem \ref{lem:(Perfect-CSIT-Throughput):}.

\subsection{\label{sub:Proof-for-Corollary-Scaling_Power}Proof for Corollary
1 (Scaling Law with Transmit SNR)}

With (\ref{eq:throughput_under_dynamic}) and by Jensen's inequality,
we can further get
\begin{eqnarray*}
R_{lim} & \geq & \sum_{j=1}^{K}d\cdot\int_{0}^{+\infty}\log\left(1+\frac{P}{d}\cdot v\right)f(v)\textrm{d}v-\sum_{j=1}^{K}d\cdot\log\left(1+P\sum_{i,i\neq j}^{K}\left(\frac{\beta_{ji}l_{ji}}{M_{ji}^{r}M_{ji}^{t}-1}\right)\cdot2^{-\frac{B_{ji}^{*}}{M_{ji}^{r}M_{ji}^{t}-1}}\right)\\
 & \geq & R_{per}-Kd\cdot\log\left(1+\frac{P}{K}\cdot\sum_{i.j,i\neq j}^{K}\left(\frac{\beta_{ji}l_{ji}}{M_{ji}^{r}M_{ji}^{t}-1}\right)\cdot2^{-\frac{B_{ji}^{*}}{M_{ji}^{r}M_{ji}^{t}-1}}\right).
\end{eqnarray*}
As $\lim_{P\rightarrow\infty}\frac{R_{per}}{\log P}=Kd$, we get that
the sum DoFs of the system are kept if
\begin{equation}
\frac{P}{K}\cdot\sum_{i.j,i\neq j}^{K}\left(\frac{\beta_{ji}l_{ji}}{M_{ji}^{r}M_{ji}^{t}-1}\right)\cdot2^{-\frac{B_{ji}^{*}}{M_{ji}^{r}M_{ji}^{t}-1}}\leq C_{0},\label{eq:bits_scaling_eq}
\end{equation}
where $C_{0}$ is some bounded constant that does not depend on $P$.

As $P\rightarrow\infty$, $B\rightarrow\infty$, we have $b\rightarrow\infty$
in (\ref{eq:optimal_bits_allocation}). Therefore, we have that when
$l_{ji}\neq0$, $B_{ji}^{*}=(M_{ji}^{r}M_{ji}^{t}-1)b+c_{ji}$; when
$l_{ji}=0$, $B_{ji}^{*}=0$ in (\ref{eq:optimal_bits_allocation}).
Substitute these $\{B_{ji}^{*}\}$ into (\ref{eq:bits_scaling_eq}),
we obtain $b\geq\log P+C_{1}$. Hence, the sum feedback bits is
\[
B=\sum_{i.j,i\neq j}^{K}B_{ji}=\sum_{i\neq j,l_{ji}\neq0}(M_{ji}^{r}M_{ji}^{t}-1)b+C_{b}\geq\sum_{i,j,i\neq j}^{K}\left\{ I_{\{l_{ji}>0\}}\cdot(N_{r}N_{t}\rho_{ji}-1)\right\} \log P+C_{b},
\]
where $c_{ji}$, $C_{1}$, $C_{b}$ above are some bounded constants
independent of $P$.

\bibliographystyle{IEEEtran}
\bibliography{Lfia_Ref}

\begin{thebibliography}{10}
\providecommand{\url}[1]{#1}
\csname url@samestyle\endcsname
\providecommand{\newblock}{\relax}
\providecommand{\bibinfo}[2]{#2}
\providecommand{\BIBentrySTDinterwordspacing}{\spaceskip=0pt\relax}
\providecommand{\BIBentryALTinterwordstretchfactor}{4}
\providecommand{\BIBentryALTinterwordspacing}{\spaceskip=\fontdimen2\font plus
\BIBentryALTinterwordstretchfactor\fontdimen3\font minus
  \fontdimen4\font\relax}
\providecommand{\BIBforeignlanguage}[2]{{%
\expandafter\ifx\csname l@#1\endcsname\relax
\typeout{** WARNING: IEEEtran.bst: No hyphenation pattern has been}%
\typeout{** loaded for the language `#1'. Using the pattern for}%
\typeout{** the default language instead.}%
\else
\language=\csname l@#1\endcsname
\fi
#2}}
\providecommand{\BIBdecl}{\relax}
\BIBdecl

\bibitem{cadambe2008interference}
V.~Cadambe and S.~Jafar, ``Interference alignment and degrees of freedom of the
  {K}-user interference channel,'' \emph{IEEE Trans. Inf. Theory}, vol.~54,
  no.~8, pp. 3425--3441, Aug. 2008.

\bibitem{jafar2010MNDoF}
T.~Gou and S.~Jafar, ``Degrees of freedom of the {K}-user {MIMO} interference
  channel,'' \emph{IEEE Trans. Inf. Theory}, vol.~56, no.~12, pp. 6040--6057,
  Dec. 2010.

\bibitem{jafar2007MIMOX}
S.~Jafar and S.~Shamai, ``Degrees of freedom region of the {MIMO X} channel,''
  \emph{IEEE Trans. Inf. Theory}, vol.~54, no.~1, pp. 151--170, Jan. 2008.

\bibitem{gomadam2011distributed}
K.~Gomadam, V.~Cadambe, and S.~Jafar, ``A distributed numerical approach to
  interference alignment and applications to wireless interference networks,''
  \emph{IEEE Trans. Inf. Theory}, vol.~57, no.~6, pp. 3309--3322, June 2011.

\bibitem{peters2011cooperative}
S.~Peters and R.~Heath, ``Cooperative algorithms for {MIMO} interference
  channels,'' \emph{IEEE Trans. Veh. Technol.}, vol.~60, no.~1, pp. 206--218,
  Jan. 2011.

\bibitem{santamaria2010maximum}
I.~Santamaria, O.~Gonzalez, R.~Heath, and S.~Peters, ``Maximum sum-rate
  interference alignment algorithms for {MIMO} channels,'' in \emph{Proc. IEEE
  GLOBECOM}, Dec. 2010, pp. 1--6.

\bibitem{jindal2006mimo}
N.~Jindal, ``{MIMO} broadcast channels with finite-rate feedback,'' \emph{IEEE
  Trans. Inf. Theory}, vol.~52, no.~11, pp. 5045--5060, Nov. 2006.

\bibitem{yoo2007multi}
T.~Yoo, N.~Jindal, and A.~Goldsmith, ``Multi-antenna downlink channels with
  limited feedback and user selection,'' \emph{IEEE J. Sel. Areas Commun.},
  vol.~25, no.~7, pp. 1478--1491, Sep. 2007.

\bibitem{ayach2012interference}
O.~Ayach and R.~Heath, ``Interference alignment with analog channel state
  feedback,'' \emph{IEEE Trans. Wireless Commun.}, vol.~11, no.~2, pp.
  626--636, Feb. 2012.

\bibitem{kim2012new}
J.-S. Kim, S.-H. Moon, S.-R. Lee, and I.~Lee, ``A new channel quantization
  strategy for {MIMO} interference alignment with limited feedback,''
  \emph{IEEE Trans. Wireless Commun.}, vol.~11, no.~1, pp. 358--366, Jan. 2012.

\bibitem{krishnamachari2009interference}
R.~Krishnamachari and M.~Varanasi, ``Interference alignment under limited
  feedback for {MIMO} interference channels,'' in \emph{Proc. IEEE Int. Symp.
  Information Theory (ISIT)}, June 2010, pp. 619--623.

\bibitem{thukral2009interference}
H.~Bolcskei and I.~Thukral, ``Interference alignment with limited feedback,''
  in \emph{Proc. IEEE Int. Symp. Information Theory (ISIT)}, July 2009, pp.
  1759--1763.

\bibitem{love2004value}
D.~Love, J.~Heath, R.W., W.~Santipach, and M.~Honig, ``What is the value of
  limited feedback for {MIMO} channels?'' \emph{IEEE Commun. Mag.}, vol.~42,
  no.~10, pp. 54--59, Oct. 2004.

\bibitem{jongren2002combining}
G.~Jongren, M.~Skoglund, and B.~Ottersten, ``Combining beamforming and
  orthogonal space-time block coding,'' \emph{IEEE Trans. Inf. Theory},
  vol.~48, no.~3, pp. 611--627, Mar. 2002.

\bibitem{shiu2000fading}
D.-S. Shiu, G.~Foschini, M.~Gans, and J.~Kahn, ``Fading correlation and its
  effect on the capacity of multielement antenna systems,'' \emph{IEEE Trans.
  Commun.}, vol.~48, no.~3, pp. 502--513, Mar. 2000.

\bibitem{ruan2011dynamic}
L.~Ruan and V.~Lau, ``Dynamic interference mitigation for generalized partially
  connected quasi-static {MIMO} interference channel,'' \emph{IEEE Trans.
  Signal Process.}, vol.~59, no.~8, pp. 3788--3798, Aug. 2011.

\bibitem{dai2008quantization}
W.~Dai, Y.~Liu, and B.~Rider, ``Quantization bounds on grassmann manifolds and
  applications to {MIMO} communications,'' \emph{IEEE Trans. Inf. Theory},
  vol.~54, no.~3, pp. 1108--1123, Mar. 2008.

\bibitem{raghavan2007systematic}
V.~Raghavan, R.~Heath, and A.~Sayeed~M., ``Systematic codebook designs for
  quantized beamforming in correlated {MIMO} channels,'' \emph{IEEE J. Sel.
  Areas Commun.}, vol.~25, no.~7, pp. 1298--1310, Sep. 2007.

\bibitem{raghavan2006near}
V.~Raghavan, A.~Sayeed, and N.~Boston, ``Near-optimal codebook constructions
  for limited feedback beamforming in correlated {MIMO} channels with few
  antennas,'' in \emph{Proc. IEEE Int. Symp. Information Theory (ISIT)}, July
  2006, pp. 2622--2626.

\bibitem{bhagavatula2011adaptive}
R.~Bhagavatula and R.~Heath, ``Adaptive limited feedback for sum-rate
  maximizing beamforming in cooperative multicell systems,'' \emph{IEEE Trans.
  Signal Process.}, vol.~59, no.~2, pp. 800--811, Feb. 2011.

\bibitem{cho2011feedback}
\BIBentryALTinterwordspacing
S.~Cho, H.~Chae, K.~Huang, D.~Kim, V.~Lau, H.~Seo, and B.~Kim,
  ``Feedback-topology designs for interference alignment in {MIMO} interference
  channels,'' \emph{submitted to IEEE Trans. Sig. Process.}, 2011. [Online].
  Available: \url{http://arxiv.org/abs/1105.5476}
\BIBentrySTDinterwordspacing

\bibitem{tulino2004random}
A.~Tulino and S.~Verd{\'u}, ``Random matrix theory and wireless
  communications,'' \emph{Foundations and Trends in Commun. and Inf. Theory},
  vol.~1, no.~1, pp. 1--182, 2004.

\bibitem{roh2006transmit}
J.~Roh and B.~Rao, ``Transmit beamforming in multiple-antenna systems with
  finite rate feedback: a {VQ}-based approach,'' \emph{IEEE Trans. Inf.
  Theory}, vol.~52, no.~3, pp. 1101--1112, Mar. 2006.

\bibitem{loyka2001channel}
S.~Loyka, ``Channel capacity of {MIMO} architecture using the exponential
  correlation matrix,'' \emph{IEEE Commun. Lett.}, vol.~5, no.~9, pp. 369--371,
  Sep. 2001.

\bibitem{love2006limited}
D.~Love and J.~Heath, R.W., ``Limited feedback diversity techniques for
  correlated channels,'' \emph{IEEE Trans. Veh. Technol.}, vol.~55, no.~2, pp.
  718--722, Mar. 2006.

\bibitem{xia2006design}
P.~Xia and G.~Giannakis, ``Design and analysis of transmit-beamforming based on
  limited-rate feedback,'' \emph{IEEE Trans. Signal Process.}, vol.~54, no.~5,
  pp. 1853--1863, May 2006.

\bibitem{zheng2008analysis}
\BIBentryALTinterwordspacing
J.~Zheng and B.~Rao, ``Analysis of vector quantizers using transformed
  codebooks with application to feedback-based multiple antenna systems,''
  \emph{EURASIP Journal on Wireless Communications and Networking}, vol. 2008,
  no.~1, p. 125892, 2008. [Online]. Available:
  \url{http://jwcn.eurasipjournals.com/content/2008/1/125892}
\BIBentrySTDinterwordspacing

\bibitem{zheng2007analysis}
------, ``Analysis of multiple antenna systems with finite-rate channel
  information feedback over spatially correlated fading channels,'' \emph{IEEE
  Trans. Signal Process.}, vol.~55, no.~9, pp. 4612--4626, Sep. 2007.

\bibitem{bernstein2011matrix}
D.~Bernstein, \emph{Matrix mathematics: theory, facts, and formulas}.\hskip 1em
  plus 0.5em minus 0.4em\relax Princeton University Press, 2011.

\bibitem{zheng2007analysisDetailed}
J.~Zheng, E.~Duni, and B.~Rao, ``Analysis of multiple-antenna systems with
  finite-rate feedback using high-resolution quantization theory,'' \emph{IEEE
  Trans. Signal Process.}, vol.~55, no.~4, pp. 1461--1476, April 2007.

\bibitem{hjorungnes2007complex}
A.~Hjorungnes and D.~Gesbert, ``Complex-valued matrix differentiation:
  Techniques and key results,'' \emph{IEEE Trans. Signal Process.}, vol.~55,
  no.~6, pp. 2740--2746, June 2007.

\end{thebibliography}

\end{document}